\documentclass{article}
\usepackage{amsmath,amssymb,mathrsfs}
\usepackage{graphics}
\usepackage{graphicx}
\usepackage{titlesec}
\usepackage{subfigure}
\graphicspath{{figures/png/}}
\usepackage{amsfonts}
\usepackage{latexsym}
\usepackage{url}
\usepackage{ifthen}  
\usepackage{bbm}

\newtheorem{theorem}{Theorem}
\newtheorem{corollary}{Corollary}
\newtheorem{lemma}{Lemma}%[section]
\newenvironment{proof}
{\noindent {\bf Proof.}}
{$\blacksquare$}
\renewcommand{\>}{\rangle}
\newcommand{\<}{\langle} 
\newcommand{\ip}{{\bf i'}}
\newcommand{\jp}{{\bf j'}}

\titleformat{\paragraph}
{\normalfont\normalsize\bfseries}{\theparagraph}{1em}{}
\titlespacing*{\paragraph}
{0pt}{3.25ex plus 1ex minus .2ex}{1.5ex plus .2ex}
\begin{document}
\title{Quantum networks: Anti-core of spin chains}
\author{E. Jonckheere, F. C. Langbein and S. Schirmer}
\maketitle

\begin{abstract}
The purpose of this paper is to exhibit a quantum network \linebreak
phenomenon---the anti-core---that goes against the classical network
concept of congestion core.  Classical networks idealized as infinite,
Gromov hyperbolic spaces with least-cost path routing (and subject to a
technical condition on the Gromov boundary) have a congestion core,
defined as a subnetwork that routing paths have a high probability of
visiting.  Here, we consider quantum networks, more specifically spin
chains, define the so-called maximum excitation transfer probability
$p_{\max}(i,j)$ between spin $i$ and spin $j$, and show that the
central spin has among all other spins the lowest probability of being
excited or transmitting its excitation.  The
anti-core is singled out by analytical formulas for
$p_{\mathrm{max}}(i,j)$, revealing the number theoretic properties of
quantum chains. By engineering the chain, we
further show that this probability can be made vanishingly small.  

\end{abstract}

\section{Introduction}

Probably the most significant result of the Gromov analysis of classical
networks~\cite{4_point,scaled_gromov} is existence of a congestion core.  Under a network protocol
that sends the packets along least cost paths, the {\it core} can be qualitatively
defined as a point where most of the geodesics (least cost paths)
converge, creating packet drops, high retransmission rates, and other
nuisances under the TCP-IP protocol~\cite{mingji_thesis}.  Existence of the
core has been experimentally observed~\cite{arXiv_dmitri} and
mathematically proved~\cite{internet_mathematics} if the network is
Gromov hyperbolic, subject to some highly technical conditions related
to the Gromov boundary~\cite{baryshnikov_tucci}.  A Gromov hyperbolic
network can intuitively be defined as a network that ``looks like'' a
negatively curved Riemannian manifold (e.g., a saddle) when viewed from
a distance. See, e.g.,~\cite{BridsonHaefliger1999} for a precise definition. 

Next to classical networks, one can envision quantum networks: the nodes
are spins that can be up $\left\vert\uparrow\right\>$ (not excited) or
down $\left|\downarrow \right\>$ (excited) and the links are quantum
mechanical couplings of the XX or Heisenberg type.  Given some random
source-destination pair $(i,j)$, a valid question is whether some spin
$\omega$ could act as a ``core,'' that is, a spin that could be excited
no matter what the source and the destination are.  For a linear chain,
one would expect such a congestion core in the center as classically any
excitation in one half of the chain would have to transit the center of
the chain to reach the other half.  In this work, we demonstrate that
quantum-mechanically the transmission of excitations does not need to
occur this way, and in fact the center $\omega$ of an odd-length spin
chain can act as an ``anti-core,'' excitation of which is avoided.

This ``anti-core," or ``anti-gravity'' center as it was originally
called, was first observed in~\cite{first_with_sophie}.  The anti-core
$\omega$ was defined as a point of high inertia, $\sum_i
d^\alpha(i,\omega)$, $\alpha \geq 1$, as opposed to the classical
congestion core that has minimum inertia owing to the negative curvature
of the underlying space~\cite[Theorem 3.2.1]{Jost1997}.  The inertia
quantifies how difficult communication to and from the anti-core is.

As it has been done along this line of work, a pre-metric
$d(\cdot,\cdot)$ based on the {\it Information Transfer Capacity (ITC)}
(see Sec.~\ref{s:homogeneous}) is employed.  Unlike standard quantum
mechanical distances~\cite{polarized_distance,arccos},~\cite[pages
412-413]{book_nielsen}, this ``distance" measure aims to quantify not
how distant two fixed quantum states are, but how close to a desired
target state a quantum state can get under the evolution of a particular
Hamiltonian.  The initial and target states are typically orthogonal.

In this paper, we provide an {\it analytical} justification of the
numerically observed anti-core phenomenon in spin chains with XX coupling, 
starting with finite-length chains, extending the ITC concept to semi- and
bi-infinite cases (Sec.~\ref{s:infinite_chains}), 
and finally {\it proving} that $d(\omega,j) \geq d(i,j)$, $\forall j \ne \omega$ in Sec.~\ref{s:anti_core}.  
We further show that by adding a bias on the central
spin its {\it ``anti-core''} property can be made stronger in the
sense that the probability of transmission of the excitation to and from
it is infinitesimally small (Sec.~\ref{s:engineered}).  

The remaining nagging question is why was it observed
in~\cite{first_with_sophie} that spin chains appear Gromov-hyperbolic
and have an anti-core, while classical networks are Gromov hyperbolic
with the opposite core?  This will be clarified in
Sec.~\ref{s:engineered} by means of a spin chain example, showing that
its Gromov boundary has {\it only one} point, while classical networks
need to have {\it at least two points} in their Gromov boundary for the
core to emerge.

\section{Metrization of homogeneous spin chains}
\label{s:homogeneous}

We consider a linear array of two-level systems (spin $\tfrac{1}{2}$
particles) with uniform coupling between adjacent spins (homogeneous
spin chain) made up of an odd number $N$ of physically equally spaced
spins with coupling Hamiltonian
\begin{equation*}
      H = \sum_{i=1}^{N-1} 
      \left(\sigma^x_i\sigma^x_{i+1} + \sigma^y_i\sigma^y_{i+1}
                             + \epsilon\sigma^z_i\sigma^z_{i+1} \right).
\end{equation*}
Here we shall be primarily interested in the case of XX coupling, for
which $\epsilon=0$.  The factor $\sigma^{x,y,z}_i$ is the Pauli matrix along
the $x,y,$ or $z$ direction of spin $i$ in the array, i.e.,
\begin{equation*}
  \sigma^{x,y,z}_i =
   I_{2\times 2} \otimes \ldots \otimes I_{2 \times 2} \otimes
   \sigma^{x,y,z}\otimes I_{2\times 2}\otimes \ldots \otimes I_{2 \times 2},
\end{equation*}
where the factor $\sigma^{x,y,z}$ occupies the $i$th position among the
$N$ factors and $\sigma^{x,y,z}$ is one of the single spin Pauli
operators
\begin{equation*}	   
    \sigma^x= \begin{pmatrix} 0 & 1 \\ 1 & 0 \end{pmatrix}, \quad
    \sigma^y= \begin{pmatrix} 0 & -\imath \\ \imath & 0  \end{pmatrix}, \quad
    \sigma^z= \begin{pmatrix} 1 & 0  \\ 0 & -1 \end{pmatrix}.
\end{equation*}
It is easily seen that $H$ is {\it real} and symmetric.

\subsection{Single excitation subspace }

The $2^N \times 2^N$ Hamiltonian commutes with the operator
$S=\sum_{i=1}^N \sigma_i^{z}$ which counts the total number of
excitations.  The Hilbert space can therefore be decomposed into
subspaces corresponding to the number of excitations.  Define
$|i\>=|\uparrow \cdots \uparrow \downarrow \uparrow \cdots \uparrow\>$
to be the quantum state in which the excitation is on spin $i$.  
The single excitation subspace $\mathcal{H}_1$ is spanned by $\{|i\>:
i=1,\ldots,N\}$.  Restricted to this subspace, the Hamiltonian in this
natural basis takes the form
\begin{equation*} 
 H_1 = \begin{pmatrix}
\epsilon & 1 & \ldots & 0 & 0 & 0 & \ldots & 0 & 0 \\
1 & 0 & \ldots & 0 & 0 & 0 & \ldots & 0 & 0 \\
\vdots & \vdots & \ddots & \vdots & \vdots & \vdots & & \vdots & \vdots \\
0 & 0 & \ldots & 0 & 1 & 0 & \ldots & 0 & 0 \\
0 & 0 & \ldots & 1 & 0 & 1 & \ldots & 0 & 0 \\
0 & 0 & \ldots & 0 & 1 & 0  & \ldots & 0 & 0\\
\vdots   & \vdots & & \vdots & \vdots & \vdots & \ddots & \vdots & \vdots \\
0 & 0 & \ldots & 0 & 0 & 0 & \ldots & 0 & 1 \\
0 & 0 & \ldots & 0 & 0 & 0 & \ldots & 1 & \epsilon
\end{pmatrix}. 
\end{equation*}
For XX coupling ($\epsilon=0$), $H_1$ becomes the $N \times N$ Toeplitz
matrix $T_N$ made up of zeros on the diagonal, ones on the super- and
subdiagonal and zeros everywhere else.  Table~\ref{t:eigenstructure}
shows the eigenvalues and eigenvectors of $H_1$. 

%********************************
\begin{table}
\begin{center}
\begin{tabular}{|r|c|c|}\hline\hline
 & $\lambda_k$ & $v_{kj}$ \\ \hline\hline 
&& \\
XX coupling ($\epsilon=0$) & $2\cos\frac{\pi k}{N+1}$ & $\sqrt{\frac{2}{N+1}}\sin\frac{\pi jk}{N+1}$ \\ 
&&\\ \hline
\end{tabular}
\end{center}
\caption{Eigenvalues and eigenvectors of Single Excitation Subspace
Hamiltonian $H_1$ under XX coupling~\cite{sophie_long}.}  \label{t:eigenstructure}
\end{table}

\subsection{Information Transfer Capacity (ITC) semi-metric}

The probability for the system to transfer from state $|i\>$ to state
$|j\>$ in an amount of time $t$, that is, the probability of transfer of
the excitation from spin $i$ to spin $j$ in an amount of time $t$,
is
\begin{equation*} 
 p_t(i,j)=|\<i|e^{-\imath H_1 t}|j\>|^2.
\end{equation*}
Observe that $\sum_{j=1}^Np_t(i,j)=1$.  In order to remove the
dependency of the probability distribution on the time, we proceed as in
Refs~\cite{quantum_rome,first_with_sophie}:
\begin{subequations}
\label{e:pmax}
\begin{align}
p_t(i,j)&=    \left|\sum_{k=1}^N e^{-\imath \lambda_k t}
 \<i|v_k\>\<v_k|j\>\right|^2  \\
        &\leq \left|\sum_{k=1}^N |\<i|v_k\>\<v_k|j\>|\right|^2 =:p_{\mathrm{max}}(i,j).
\end{align}
\end{subequations}
We refer to $p_{\mathrm{max}}(i,j)$ as {\it maximum transfer
probability} from $|i\>$ to $|j\>$ or {\it Information Transfer Capacity
(ITC)} between $|i\>$ and $|j\>$.  Its explicit formulation for XX
chains is easily obtained from~\eqref{e:pmax} and
Table~\ref{t:eigenstructure}:
\begin{equation}
\label{e:explicitpmax} 
\sqrt{p_{\mathrm{max}}(i,j)}=
\frac{2}{N+1}\sum_{k=1}^{2n+1} 
  \left|\sin\frac{\pi k i}{2(n+1)}\right|\left|\sin\frac{\pi k j}{2(n+1)}\right|.
\end{equation}

\begin{lemma}
\label{l:pmaxii=1}
$p_{\mathrm{max}}(i,j)\leq1$ and $p_{\mathrm{max}}(i,i)=1$ for all $i,j=1,\ldots,N$.
\end{lemma}

\begin{proof}
$p_{\max}(i,i)=1$ follows directly from (\ref{e:pmax}a), setting $i=j$
and $t=0$ and noting that the eigenvectors $|v_m\>$ form an orthonormal
basis.  $p_{\rm max}(i,j)\leq 1$ then follows from a Cauchy-Schwartz
argument.
\end{proof}

The preceding lemma tells us that in order to define a (pre)metric from
$p_{\mathrm{max}}(i,j)$, it is legitimate to take the $\log$ and define
\begin{equation} 
      d(i,j):=-\log p_{\mathrm{max}}(i,j)
\end{equation}
on the single excitation subspace of the chain.  From
Lemma~\ref{l:pmaxii=1}, $d(i,i)=0$ and $d(i,j) \geq 0$, and clearly
$d(i,j)=d(j,i)$.  Observe, however, that $d(i,j)$ can vanish for $i\neq j$
and the triangle inequality need not be satisfied, so that $d(i,j)$ is
just a pre-metric, but this will be sufficient for our purposes.

The definition of $d(i,j)$ bears some commonality with sensor
networks~\cite{eurasip_clustering}, where the Packet Reception Rate
$\mathrm{PRR}(i,j)$ from sensor $i$ to sensor $j$---that is, the
probability of successful transmission of packets from $i$ to
$j$---defines a premetric $d(i,j)=-\log \mathrm{PRR}(i,j)$.

\section{Infinite chains}
\label{s:infinite_chains}

In this section, we develop some asymptotic formulas for
$\sqrt{p_{\mathrm{max}}(i,j)}$ for infinite-length chains in order to show that
the central spin $n+1$ of a chain of odd length $N=2n+1$ has the lowest
probability of being excited, hence justifying the terminology of
``anti-core,'' even for $N \to \infty$.  
This will further reveal a classical-quantum discrepancy: 
Classical dynamical systems
interconnected in an homogeneous infinite chain architecture exhibit the
so-called {\it shift-invariance,} that is, those dynamical interactions
depending on the positions $i$ and $j$ of two systems in the chain in
fact depend only on the distance $|i-j|$.  As a corollary of the
asymptotic formulas, this well known shift-invariance does not carry
over to the quantum chains---no matter how the chain is extended to
infinity, two spins in their transfer probability interaction keep
properties specific to some number theoretic properties of their
positions $i$ and $j$.  Moreover, in a classical chain, the interaction
at infinity is insensitive to the way the limit is taken: either the
chain starts at a specific system, say $1$, and extends to infinity as
\[
  (1,2,3,...), \quad (\mbox{``semi-infinite chain,'' written }\rightarrow) 
\]
or the chain starts at its center $\omega$ and extends both ways as 
\[ 
   (...,\omega-2,\omega-1,\omega,\omega + 1, \omega+2,...), 
   \quad \mbox{(``doubly-infinite chain'' written} \leftrightarrow).
\] 
It is another quantum mechanical effect that the two infinite chains do
not yield the same asymptotic transfer probabilities.

\subsection{Semi-infinite chains}

\begin{theorem}
\label{t:semi_infinite_chains}
For a semi-infinite XX chain, the maximum transition probabilities are
given by
\begin{align*}
%\label{e:series_representation}
\sqrt{p_{\mathrm{max}}^{\rightarrow}(i,j)}
&= \frac{4}{\pi^2}\left( 2+\sum_{m=2,4,...} \frac{4}{(m^2{\bf j}^2-1)(m^2{\bf i}^2-1)} \right)\\
&= \frac{8}{\pi^2} \left( 
   \frac{\bf{i}^2}{\bf{i}^2-\bf{j}^2} \left(\frac{\pi}{2\bf{i}} \right)
                                       \cot\left(\frac{\pi}{2\bf{i}} \right)
    -\frac{\bf{j}^2}{\bf{i}^2-\bf{j}^2} \left(\frac{\pi}{2\bf{j}}\right)
                                        \cot\left(\frac{\pi}{2\bf{j}} \right)\right), 
\end{align*}
where ${\bf{i}}=i/\mathrm{gcd}(i,j)$ and ${\bf{j}}=j/\mathrm{gcd}(i,j)$, 
and $\mathrm{gcd}(i,j)$ denotes the greatest common divisor of $i$ and $j$. 
\end{theorem}

\begin{proof}
The proof is in Appendix~\ref{s:semi_infinite_chains}.
\end{proof}

Lemma~\ref{l:pmaxii=1} provided some ``probability'' interpretations of 
$p_{\max}(i,j)$ for finite chains.  We show that the same interpretation
holds for infinite chains.

\begin{lemma} 
$p_{\mathrm{max}}^{\rightarrow}(i,i)=1$ and
$p_{\mathrm{max}}^{\rightarrow}(i,j)<1$ for $i \ne j$. 
\end{lemma}

\begin{proof}
For $i=j$, $\mathrm{gcd}(i,j)=i=j$, so that ${\bf i}={\bf j}=1$ and
it remains to show that
\begin{equation*} 
  \frac{4}{\pi^2}\left( 2 + \sum_{m=2}^{\infty} \frac{4}{(m^2-1)^2} \right)=1.
\end{equation*}
This can be derived as follows.  From the definition of the Riemann
$\zeta$ function, the following is easily verified:  
\begin{equation*}
  2 \sum_{m=2,4,...}^\infty \frac{1}{m^s}=2^{1-s}\zeta(s).
\end{equation*}
Observing that the left-hand side is
$2\left(\zeta(s)-\sum_{m=1,3,...}\frac{1}{m^s}\right)$, it follows that
\begin{equation*}
  \zeta(s)(1-2^{-s})=\sum_{\mu=1}^\infty\frac{1}{(2\mu-1)^s}.
\end{equation*}
Setting $s=2$ and remembering that $\zeta(2)=\pi^2/6$ (Euler formula) give
\begin{equation*}  
  \sum_{\mu=1}^\infty\frac{1}{(2\mu-1)^2}=\frac{\pi^2}{8}.
\end{equation*}
Therefore, 
\begin{equation*} 
  \sum_{m=2,4,...}\frac{16}{(m^2-1)^2}
   =4\left(\frac{\pi^2}{8}+\left(\frac{\pi^2}{8}-1\right)-1\right)=\pi^2-8.
\end{equation*}
The above and the infinite series representation yields 
\begin{equation*} 
  \sqrt{p_{\mathrm{max}}^{\rightarrow}(i,j)} 
< \frac{4}{\pi^2}\left( 2+ \sum_{m=2,4,...} \frac{4}{(m^2-1)^2}\right)=1.
\end{equation*}
\end{proof}

It is interesting to observe from the infinite series representation that \linebreak
$p_{\mathrm{max}}^{\rightarrow}(i,j)$ dips when $i$ and $j$ are
relatively prime.  In particular, relative to the anti-core $j=n+1$, the
deepest dips happen at $i=1$ and $i=2n+1$, since
$\mathrm{gcd}(1,n+1)=1$ and $\mathrm{gcd}(n+1,2n+1)=1$. 
The opposite phenomenon happens when $i$ and $j$ share prime factors. 
In this case, ${\bf i}$ and ${\bf j}$ drop, 
hence by the infinite series representation $p_{\mathrm{max}}(i,j)$ shoots up. 
This explains the ``ripples'' in the $\sqrt{p_{\mathrm{max}}(i,j)}$ plots of Fig.~\ref{f:dramatic_anti_core}. 
Even though this figure is the case of a finite length chain, 
the ``ripple'' phenomenon is well explained by the asymptotic formula. 

As a word of warning, the ``spikes'' near the anti-core of
Fig.~\ref{f:dramatic_anti_core} {\it should not} be misconstrued as
``cores.'' Indeed, the first spike occurs at $j=87$, so all it is
depicting is the trivial fact that $p_{\mathrm{max}}(87,87)=1$; this
implies, by mirror symmetry relative to the middle spin, that
$p_{\mathrm{max}}(87,115)=1$ as well.

\begin{corollary}
\label{c:semi_infinite_finite_diameter} The diameter of the
semi-infinite chain is finite and is achieved along a sequence $\{i_{k
\in \mathbb{N}}\}$ of prime numbers such that $\lim_{k \to \infty}
i_k=\infty$. 
\end{corollary}

\begin{proof}
From the infinite series representation, it is clear that 
$p^{\rightarrow}_{\mathrm{max}}(i,j)\geq 64/\pi^4$. 
Hence $\sup_{i\ne j}d^{\rightarrow}(i,j)\leq -\log(64/\pi^4)$. 
To show that this can be achieved, it suffices to observe
that the infinite series goes to $0$ along an $i$-sequence (or
$j$-sequence) of prime numbers.
\end{proof}

\subsection{Doubly-infinite chains}

In the doubly infinite chain case, the position of the spins is referenced to
$\omega$.  Hence, define $i'=i-\omega$ and $j'=j-\omega$. Furthermore,
$\ip=i'/\mathrm{gcd}(i',j')$ and $\jp=j'/\mathrm{gcd}(i',j')$.

\begin{theorem}
\label{t:doubly_infinite_chains_same}
Consider an homogeneous XX chain of odd length $N=2n+1$  
with the positions $i'$, $j'$ of the spins referenced relative to the center $n+1$. 
Assume $i'$ and $j'$ are positive. 

If both $i'$ and $j'$ are odd or both $i'$ and $j'$ are even with the
same power of $2$ in their prime number factorization, we have 
\begin{align*}
\sqrt{p_{\mathrm{max}}^\leftrightarrow (i',j')}
&= \frac{4}{\pi^2}\left( 2+\sum_{m=2,4,...}\frac{4}{(m^2\ip^2-1)(m^2 \jp^2-1)} \right)\\
&=\frac{8}{\pi^2}
\left(
\frac{1}{\ip^2-\jp^2}
\left( 
    \ip^2 \left(\frac{\pi}{2\ip}\right)  \cot \left(\frac {\pi}{2 \ip}\right)
  - \jp^2 \left(\frac{\pi}{2\jp}\right)  \cot \left(\frac {\pi}{2 \jp}\right)  
\right) 
\right).
\end{align*}

If $i'$ and $j'$ are even with different powers of $2$ in their prime
number factorization or $i'$ is odd and $j'$ is even,
\begin{align*}
\sqrt{p_{\mathrm{max}}^\leftrightarrow(i',j')}
&= \frac{4}{\pi^2} \left( 2+\sum_{m=4,8,...} \frac{4}{(m^2\ip^2-1)(m^2\jp^2-1)}\right) \\
&=\frac{8}{\pi^2} \left(\frac{\ip^2}{\ip^2-\jp^2} 
    \left(\frac{\pi}{4\ip}\right)  \cot\left( \frac{\pi}{4\ip}\right)
        - \frac{\jp^2}{\ip^2-\jp^2} \left( \frac{\pi}{4\jp}\right)  
                                   \cot\left( \frac{\pi}{4\jp}\right) \right).
\end{align*}
\end{theorem}

\begin{proof}
See Appendix~\ref{s:doubly_infinite_chains}. 
\end{proof}

\begin{theorem}
\label{t:doubly_infinite_chains_0} For a homogeneous XX chain of length
$N=2n+1$ with the positions $0$, $j'$ of the spins referenced relative
to the center $n+1$ and $j' > 0$
\begin{equation} 
    \sqrt{p_{\mathrm{max}}^{\leftrightarrow}(0,j')}=\frac{2}{\pi}\approx 0.636619. 
\end{equation}
\end{theorem}

\begin{proof}
The result follows from the integral formulas of
Sec.~\ref{s:consistency} of Appendix~\ref{s:doubly_infinite_chains}.
\end{proof}

As a corollary of this theorem, we show that its asymptotic formula
predicts the magnitude of the dip of Figure~\ref{f:dramatic_anti_core}.
Observe the following:
\begin{align*}
&\sqrt{p_{\mathrm{max}}^{[1:201]}(87,101)} \\
&= \sqrt{p_{\mathrm{max}}^{[1:201]}(101,87)} \quad (\mbox{by symmetry of the }p_{\mathrm{max}}\mbox{ function}) \\
&=\sqrt{p_{\mathrm{max}}^{[1:201]}(101,115)} \quad (\mbox{by mirror symmetry of chain relative to center})\\
&\approx 0.63\quad (\mbox{by inspection of Fig.}~\ref{f:dramatic_anti_core}).
\end{align*}
Next, translating the finite chain to the doubly-infinite model, one would expect
\[ \sqrt{p_{\mathrm{max}}^{[1:201]}(101,115)} \approx \sqrt{p_{\mathrm{max}}^{\leftrightarrow}(0,14)}, \]
which given the above numerical observations holds remarkably accurately. 
Although $N<\infty$, the dip value of $\sqrt{p_{\mathrm{max}}^{[1:101]}(87,101)}$ is
consistent with the asymptotic value given by
Theorem~\ref{t:doubly_infinite_chains_0}.

Observe from Theorems~\ref{t:doubly_infinite_chains_same} and
\ref{t:doubly_infinite_chains_0} that the ``probability'' interpretation
of $p_{\mathrm{max}}^{\leftrightarrow}$ holds the same way as it did for
the semi-infinite chain. The details are left out.

\begin{corollary}
The diameter of the doubly-infinite chain is finite and is achieved for
$d^{\leftrightarrow}(0,j')=-2\log (2/\pi)$.
\end{corollary}

\begin{proof}
It is easily seen from the infinite series representations of
$\sqrt{p_{\mathrm{max}}^{\leftrightarrow}(i',j')}$ in both cases of
Theorem~\ref{t:doubly_infinite_chains_same} that
$\sqrt{p_{\mathrm{max}}^{\leftrightarrow}(i'j')}\geq 8/\pi^2$, $\forall
i', j' \ne 0$.  This together with
Theorem~\ref{t:doubly_infinite_chains_0} implies that the diameter is
finite.  Furthermore, observe that the bound $8/\pi^2$ is reached along
an infinite sequence $\{i'_{k\in \mathbb{N}}\}$ of prime numbers such
that $\lim_{k \to \infty} i'_k =\infty$, which guarantees that
$\ip_{k}:=i'_k/\mathrm{gcd}(i'_k,j') \to \infty$ at infinity.  This
together with $2/\pi < 8/\pi^2$ implies that the diameter is $-2 \log
(2/\pi)$.
\end{proof}

The fact that the diameter is achieved for one spin at $i'=0$ reveals
the ``anti-core.''

\section{Anti-core}
\label{s:anti_core}

\subsection{Minimum probability}

Inspired from congestion phenomena in classical
communications~\cite{internet_mathematics}, it was numerically observed
in~\cite{first_with_sophie} that for chains of odd length $N=2n+1$ the
inertia of the quantum network relative to the spin $j$,
$I^{(\alpha)}(j):=\sum_{i=1}^N d^\alpha(i,j)$, $\alpha =2$, is maximal
for $j=\omega:=n+1$.  We now show that a stronger result holds:
\begin{equation*}
 \arg \max_j d(i,j) = \omega, \quad \forall i \ne \omega,
\end{equation*}
In other words, for each spin other than the center, the center is the farthest
away, which of course implies that $I^{(\alpha)}(j)$ is maximum for
$j=\omega$.  The preceding can be rephrased as
\begin{equation*}
  \arg \min_j p_{\mathrm{max}}(i,j) = \omega, \quad \forall i \ne \omega.
\end{equation*}
Given the explicit expression for $p_{\max}(i,j)$ in
(\ref{e:explicitpmax}), the claim that $\sqrt{p_{\mathrm{max}}(i,j)}$
is achieved for $j=n+1$ amounts to proving the following:
\begin{theorem}
\label{t:anti_core}
For XX chains of odd length $N=2n+1$, we have
\begin{equation}
\label{e:min_prob}
\frac{2}{N+1}
\sum_{k=1}^{2n+1} \left| \sin \frac{\pi k i }{2(n+1)} \sin\frac{\pi k j}{2(n+1)}\right|  \geq  
\frac{2}{N+1} \sum_{k=1}^{2n+1} \left| \sin \frac{\pi k i }{2(n+1)}\sin \frac{\pi k }{2}\right| 
\end{equation}
as $n \to \infty$.
\end{theorem}
\begin{proof}
Firstly, we evaluate the asymptotic value of the right-hand side, that
is, the maximum excitation transition probability from spin $i$ to spin
$(n+1)=\omega$ (or from spin $\omega$ to spin $N$) for an infinite ($N
\to \infty$) chain with XX coupling.  From~\cite[Eq. (16)]{sophie_long}
or Table~\ref{t:eigenstructure}, we have
\[ \lim_{n \to \infty} \sqrt{p_{\mathrm{max}}(i,\omega)}=
\lim_{n\to \infty}\frac{1}{n+1} \sum_{k=1}^{2n+1} 
\left|\sin\frac{\pi k i}{2(n+1)}\sin \frac{\pi k}{2}\right|. \]
The even terms are zero and letting $k=2l+1$ we have
\[ \sqrt{p_{\mathrm{max}}(i,n+1)}= \frac{1}{n+1}\sum_{l=0}^{n} 
  \left| \sin \left( \frac{\pi i(2l+1)}{2(n+1)}\right)   \right|. \]
Setting $x=(2l+1)/(2(n+1))$, we have
\begin{align*}
\sqrt{p_{\mathrm{max}}(i,n+1)}
=\int_0^1 | \sin ( \pi i x) | dx 
= i \int_0^{1/i} \sin \pi i x dx
=\frac{2}{\pi}.
\end{align*}
Next, from the above and the infinite series representation of
Theorem~\ref{t:semi_infinite_chains}, it suffices to show that
\begin{equation*}
  \frac{4}{\pi^2}\left( 2+\sum_{m\in M} \frac{4}{(m^2{\bf j}^2-1)(m^2{\bf i}^2-1)}   \right)
  \geq \frac{2}{\pi}. 
\end{equation*}
Observe that 
\begin{equation*}
\frac{4}{\pi^2}\Big( 2+\sum_{m\in M} \frac{4}{(m^2{\bf j}^2-1)(m^2{\bf i}^2-1)}\Big)
%&\geq& \frac{4}{\pi^2}\left( 1+ \left.\frac{4}{(m^2{\bf j}^2-1)(m^2{\bf i}^2-1)}\right|_{m=0}   \right)\\
\geq \frac{8}{\pi^2} > \frac{2}{\pi}
\end{equation*}
and the Theorem is proved. 
\end{proof}

Thus we have identified spin $\omega=n+1$ as having minimal probability
of any excitation being transferred to or from it.  To put it another
way, the spin $\omega$ is maximally distant from all other spins.  We
shall call the corresponding probability amplitude ``$\omega$-small''
and the corresponding distance ``$\omega$-large.''  The
``$\omega$-small'' property is illustrated in Fig.~\ref{f:dramatic_anti_core}.

%%%%%%%%%%%%%%%%%%%%%%%%%%%%%%%%%%%%%
\begin{figure}
\begin{center}
\includegraphics[width=0.6\textwidth]{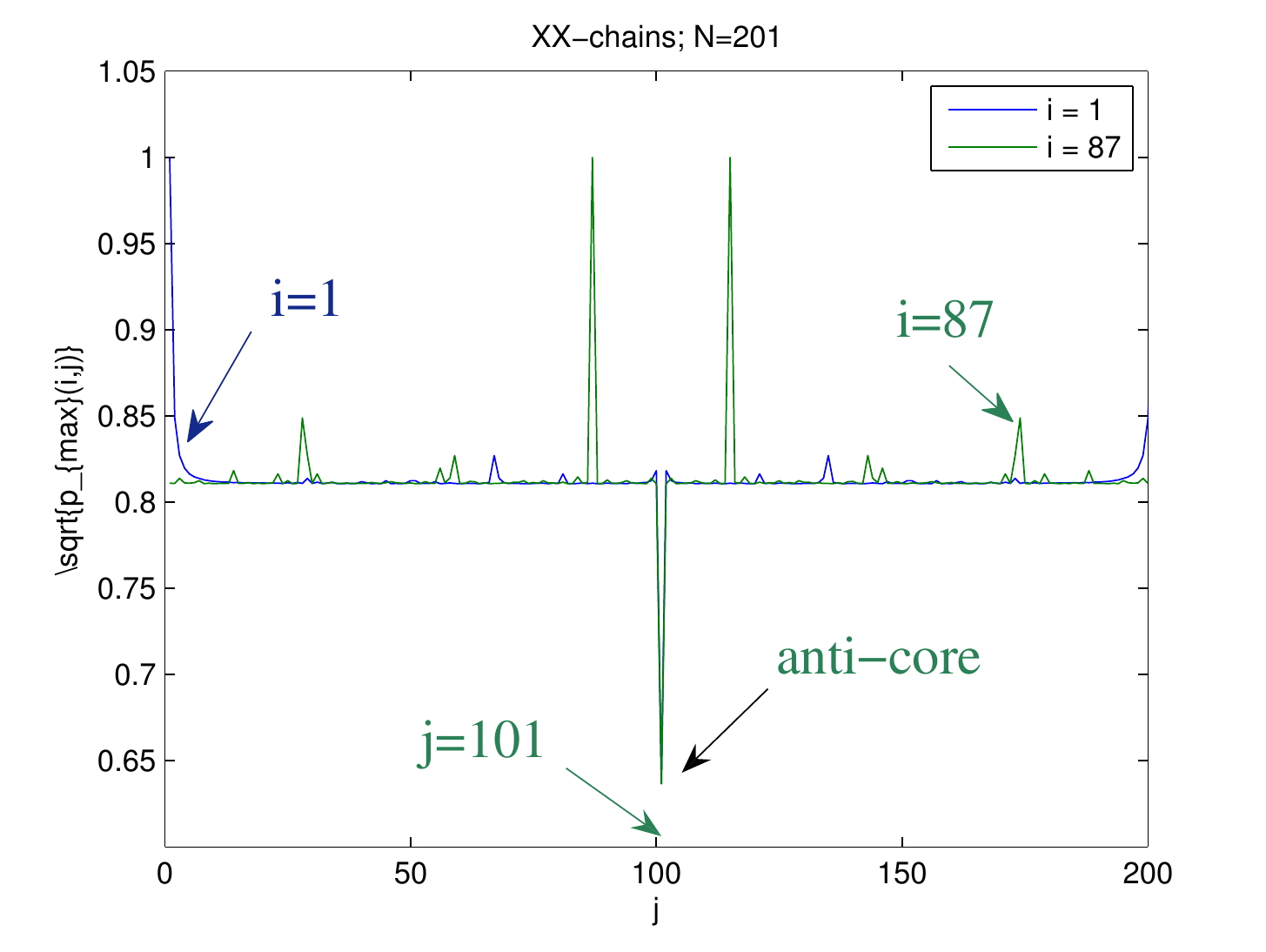}
\end{center}
\caption{Square root of maximum transition probability between spins
$i=1,87$ and all $j$-spins for an XX chain with $N=201$.  The sharp drop
at $j=101$ illustrates the ``$\omega$-small'' property. 
The two spikes near the anti-core are {\it not} congestion cores, as observed in Sec.~\ref{s:infinite_chains}. }  \label{f:dramatic_anti_core}
\end{figure}
%%%%%%%%%%%%%%%%%%%%%%%%%%%%%%%%%%%%%%%%%%

\subsection{Transport properties}

Here we examine the transport properties of the center $\omega=n+1$ and
justify its ``anti-core'' properties.  To this end, we consider the path
integral representation.  Starting with
\begin{equation*}
   \<i|e^{-\imath H_1 \tau}|k\> = 
    \sum_{\ell=1}^N \<i|e^{-\imath H_1 s} |\ell\>\<\ell| e^{-\imath H_1
    (\tau - s)}|k\>, 
\end{equation*}
we obtain
\begin{equation*}
  \<i|e^{-\imath H_1 t}|j\> 
  =  \sum_{k,\ell=1}^N \<i|e^{-\imath H_1 s} |\ell\> 
         \<\ell| e^{-\imath H_1 (\tau - s)}|k\>\<k |e^{-\imath H_1(t-\tau)}|j\>.
\end{equation*}
By iterating, we get
\begin{align*}
  \<i|e^{-\imath H_1 t}|j\> 
 &=  \sum_{k_1,k_2,\ldots,k_{n-2}=1}^N \<i|e^{-\imath H_1 t_1} |k_1\>
    \<k_1| e^{-\imath H_1 t_2 )}|k_2\>\ldots\<k_{n-2} |e^{-\imath H_1 t_{n-1}}|j\> \\
 &= \sum_{k_1,k_2,\ldots,k_{n-2}=1}^N \prod_{i=1,\ldots,n-1}\<k_{i-1}|e^{-\imath H_1 t_i}|k_i\>
\end{align*}
with $k_0=i$, $k_{n-1}=j$, $t_1+t_2+\ldots+t_{n-1}=t$, and $n \leq N$. It follows that 
\begin{align*}
\sqrt{p_t(i,j)} 
&\leq \sum_{k_1,k_2,\ldots,k_{n-2}=1}^N \prod_{i=1,\ldots,n-1}\sqrt{p_{t_i}(k_{i-1},k_i)} \\
&\leq \sum_{k_1,k_2,\ldots,k_{n-2}=1}^N \prod_{i=1,\ldots,n-1}\sqrt{p_{\mathrm{max}}(k_{i-1},k_i)}.
\end{align*}
Since the above is valid for all $t$'s, we get
\begin{equation}  
 \sqrt{p_{\mathrm{max}}(i,j)} 
 \leq \sum_{k_1,k_2,\ldots,k_{n-2}=1}^N
 \prod_{i=1,\ldots,n-1}\sqrt{p_{\mathrm{max}}(k_{i-1},k_i)}.  
\end{equation}
The above means that an excitation from the source $i$ to the
destination spin $j$ takes all possible length-$n$ paths from $i$ to
$j$, including those paths transiting through $\frac{N+1}{2}=\omega$.
For those paths, any term of the form $p_{\mathrm{max}}(\omega,k_i)$ or
$p_{\mathrm{max}}(k_{i-1},\omega)$ is $\omega$-small, making the norm of
the product in the right-hand side $\omega$-small.  Thus, for any
transfer of excitation from $i$ to $j$, the probability of exciting
$\omega$ along the way is $\omega$-small.  If we consider the
probability of excitation of $\omega$ as its ``congestion,'' then
$\omega$ remains clear of congestion, for transfer from any source $i\ne
\omega$ to any destination $j\ne \omega$.  Thus $\omega$ appears to be
the anti-thesis of the concept of core; let us agree to call it {\it ``anti-core.''}

\section{Anti-core in engineered chains}
\label{s:engineered}

As observed earlier, the diameter of a homogeneous XX chain remains
finite even as the length of the chain goes to infinity.  We now examine
whether we can modify the chain to increase its diameter to infinity.  
One way to achieve this is to apply a local potential $\zeta$ to the
central spin $\omega$.  This has the effect of perturbing the single
excitation Hamiltonian, in turn distorting the original homogeneous
distance to $d_\zeta$, 
so that in the limit $\zeta \to\infty$ the ITC diameter increases {\it ad infinitum}, 
hence getting close to the coarse geometry paradigm of
dealing with objects of infinite size. The anti-core phenomenon is
amplified in the sense that $\lim_{\zeta\to\infty}
\sqrt{p_{\mathrm{max}}(1,\omega)} =0$.  This provides a tunneling
barrier interpretation of the anti-core.

We prove that the $d_{\zeta \to \infty}$ diameter of the engineered
chain goes to $\infty$ under two different scenarios: $N<\infty$ and
$N=\infty$. The $N<\infty$ proof is in the spirit of the main body of
the paper; the $N=\infty$ proof is operator-theoretic and relies on the
assumption that $H_1$ is a doubly-infinite matrix, hence eradicating the
``border effects.''

\subsection{Finite Chains}

With the applied potential the Hamiltonian in the single excitation
subspace becomes
\begin{equation} 
H_1^{(\zeta)}= \begin{pmatrix}
0 & 1 & \ldots  & 0 & 0 & 0 & \ldots & 0 & 0 \\
1 & 0 & \ldots  & 0 & 0 & 0 & \ldots & 0 & 0 \\
\vdots & \vdots & \ddots & \vdots & \vdots & \vdots   &    & \vdots   &\vdots     \\
0 & 0 & \ldots  & 0 &  1    & 0 & \ldots  & 0 & 0  \\
0 & 0 &\ldots   & 1 & \zeta & 1 & \ldots  & 0 & 0  \\
0 & 0 & \ldots  & 0 & 1 & 0  & \ldots & 0 & 0\\
\vdots & \vdots &   & \vdots &   \vdots  & \vdots     & \ddots  & \vdots   & \vdots \\
0 & 0 & \ldots  & 0 &   0    & 0  & \ldots & 0        & 1          \\
0 & 0 & \ldots  & 0 & 0      & 0  & \ldots & 1        & 0
\end{pmatrix}.
\label{e:perturbed_H_1}
\end{equation}
The eigenstructure of this new Hamiltonian yields the new ITC distance
$d_\zeta$. 

\begin{theorem}
For an XX chain of odd length $N<\infty$, $\lim_{\zeta \to \infty}
d_\zeta(1,\omega)=\infty$.
\end{theorem}

\begin{proof}
To emphasize the dependency on the number $N$ of spins, observe that
\begin{equation*} 
 H_1^{(\zeta)}=T_N+\zeta E, 
\end{equation*}
where $E$ is the $N \times N$ matrix made up of $0$'s everywhere except
for a $1$ in position $\left(\frac{N+1}{2},\frac{N+1}{2}\right)$; and
$T_N$ is the $N \times N$ finite Toeplitz matrix made up of $0$'s on the
diagonal, $1$'s on the super-diagonal, $1$'s on the sub-diagonal, and
$0$'s everywhere else.  Recall that the determinant of the sum of two
matrices equals the sums of the determinants of all matrices made up
with some columns of one matrix and the complementary columns of the
other matrix.  Applying the latter to $\det((\lambda I_N -T_N)-\zeta
E))$ yields the characteristic polynomial
\begin{equation*}
  \det(\lambda I - T_N) -\zeta \left(\det\left(\lambda I-T_{\frac{N-1}{2}}\right)\right)^2. 
\end{equation*}
From classical root-locus techniques, it follows that, as $\zeta \to
\infty$, {\it exactly one} eigenvalue $\lambda_N$ goes to $\infty$,
while the $(N-1)$ remaining ones converge to the roots of
$\left(\det\left(\lambda I-T_{\frac{N-1}{2}}\right)\right)^2=0$.  The
eigenvector equations $(T_N+\zeta E)v_k=\lambda_k(\zeta) v_k$ split,
asymptotically as $\zeta \to \infty$, into two subsets: one for
$\lambda_N \to \infty$ and the others for $\frac{\lambda_k}{\zeta}\to
0$; that is, resp.,
\begin{align*}
 Ev_N&= \left( \lim_{\zeta \to \infty} \frac{\lambda_N(\zeta)}{\zeta}\right) v_N,\\ 
 Ev_k&= 0, \quad k\ne N.
\end{align*}
Next, again from root-locus techniques, it follows that 
\begin{equation*} 
  \lim_{\zeta \to \infty} \frac{\lambda_N(\zeta)}{\zeta}=1, 
\end{equation*}
so that $v_{N,k}=0$, $k \ne \omega$ and $v_{N,\omega}=1$. On the other hand, 
it is obvious that $v_{k,\omega}=0$, $k\ne N$. 
Therefore
\begin{equation*} 
  p_{\mathrm{max}}(1,\omega)=\sum_{k=1}^N|\<1|v_k\>\<v_k|\omega\>|=0 
\end{equation*}
and $\lim_{\zeta \to \infty}d_{\zeta}(1,\omega)=\infty$. 
\end{proof}

\subsection{Example}

To illustrate several important points, we consider a very simple example, 
which has the advantage of being analytically tractable. 
Consider the Hamiltonian~\eqref{e:perturbed_H_1} for the $N=3$ case. 
The diagonal matrix of eigenvalues of $H_1$ can be computed symbolically as
\[\begin{pmatrix}
\zeta/2 - (z^2 + 8)^{1/2}/2  &                       0 & 0\\
                       0& 0 & 0\\
                       0&                       0 & \zeta/2 + (\zeta^2 + 8)^{1/2}/2
\end{pmatrix}.\]
The normalized eigenvectors are computed as 
\[v_1=
\begin{pmatrix}
                         1/((\zeta/2 - (\zeta^2 + 8)^{1/2}/2)^2 + 2)^{1/2}\\
 (\zeta/2 - (\zeta^2 + 8)^{1/2}/2)/((\zeta/2 - (\zeta^2 + 8)^{1/2}/2)^2 + 2)^{1/2}\\
                         1/((\zeta/2 - (\zeta^2 + 8)^{1/2}/2)^2 + 2)^{1/2}
\end{pmatrix},\]
\[ v_2=\begin{pmatrix}
-2^{1/2}/2\\
          0\\
  2^{1/2}/2
\end{pmatrix},\]
\[v_3=\begin{pmatrix}
 1/((\zeta/2 + (\zeta^2 + 8)^{1/2}/2)^2 + 2)^{1/2}\\
 (\zeta/2 + (\zeta^2 + 8)^{1/2}/2)/((\zeta/2 + (\zeta^2 + 8)^{1/2}/2)^2 + 2)^{1/2}\\
                         1/((\zeta/2 + (\zeta^2 + 8)^{1/2}/2)^2 + 2)^{1/2}
\end{pmatrix}.
\]

Using those eigenvectors to symbolically compute $p_{\mathrm{max}}(1,2)$ yields
\begin{align*} 
&p_{\mathrm{max}}(1,2)=\\
&\left(\left|(\zeta/2 - (\zeta^2 + 8)^{1/2}/2)/((\zeta/2 - (\zeta^2 + 8)^{1/2}/2)^2 + 2\right|\right.\\ 
+&~~\left.\left|(\zeta/2 + (\zeta^2 + 8)^{1/2}/2)/((\zeta/2 + (\zeta^2 + 8)^{1/2}/2)^2 + 2\right|\right)^2
\sim 4 \zeta.
\end{align*}
After symbolic computation of $p_{\mathrm{max}}(1,3)$ and symbolically simplifying the expression, 
it is observed that $p_{\mathrm{max}}(1,3)=1$. The results are translated into distances and plotted in
Fig.~\ref{f:chain123}, left. 

%*************************************************
\begin{figure}
\centerline{
\subfigure{\scalebox{0.4}{\includegraphics{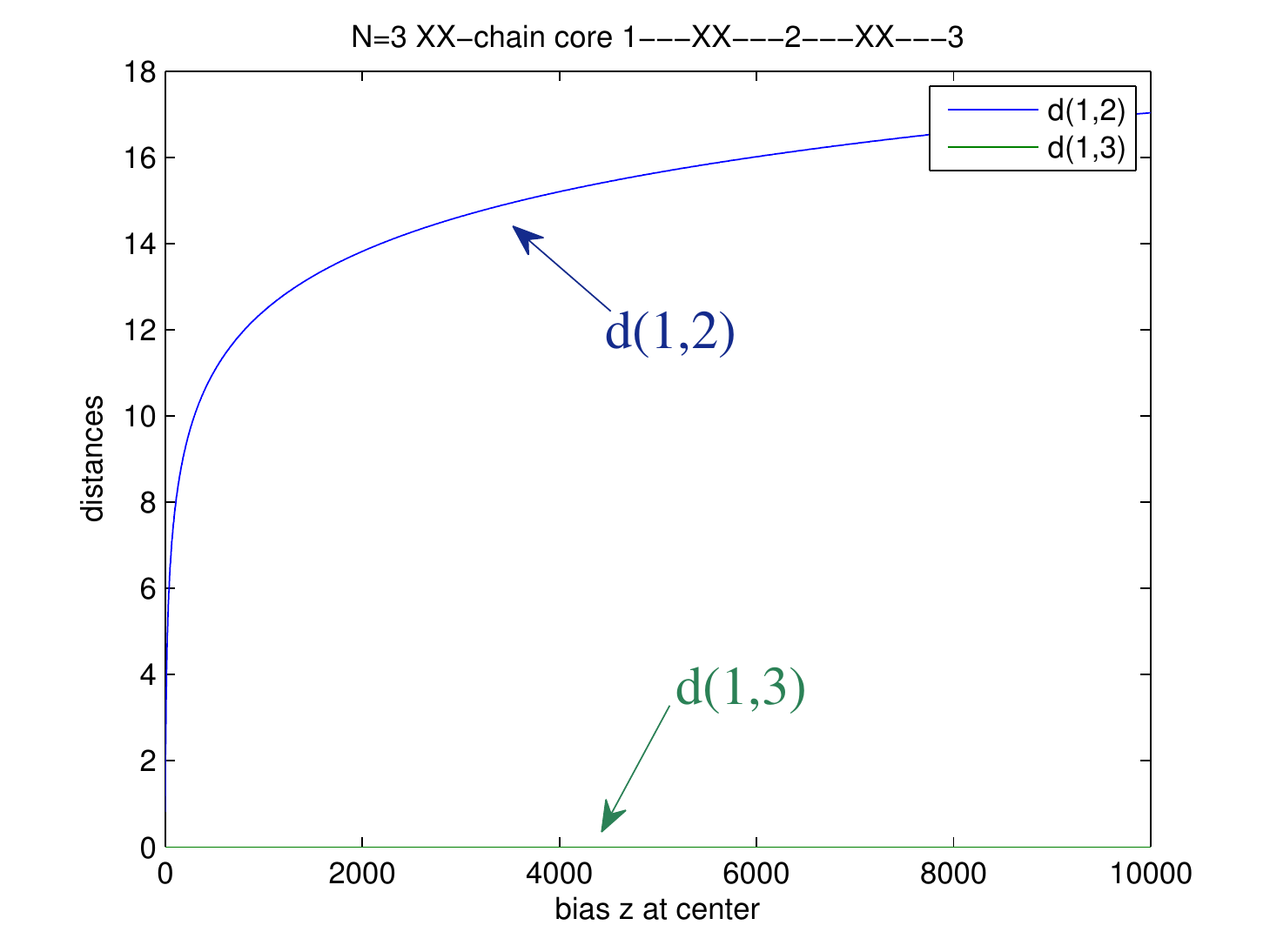}}}
\subfigure{\scalebox{0.4}{\includegraphics{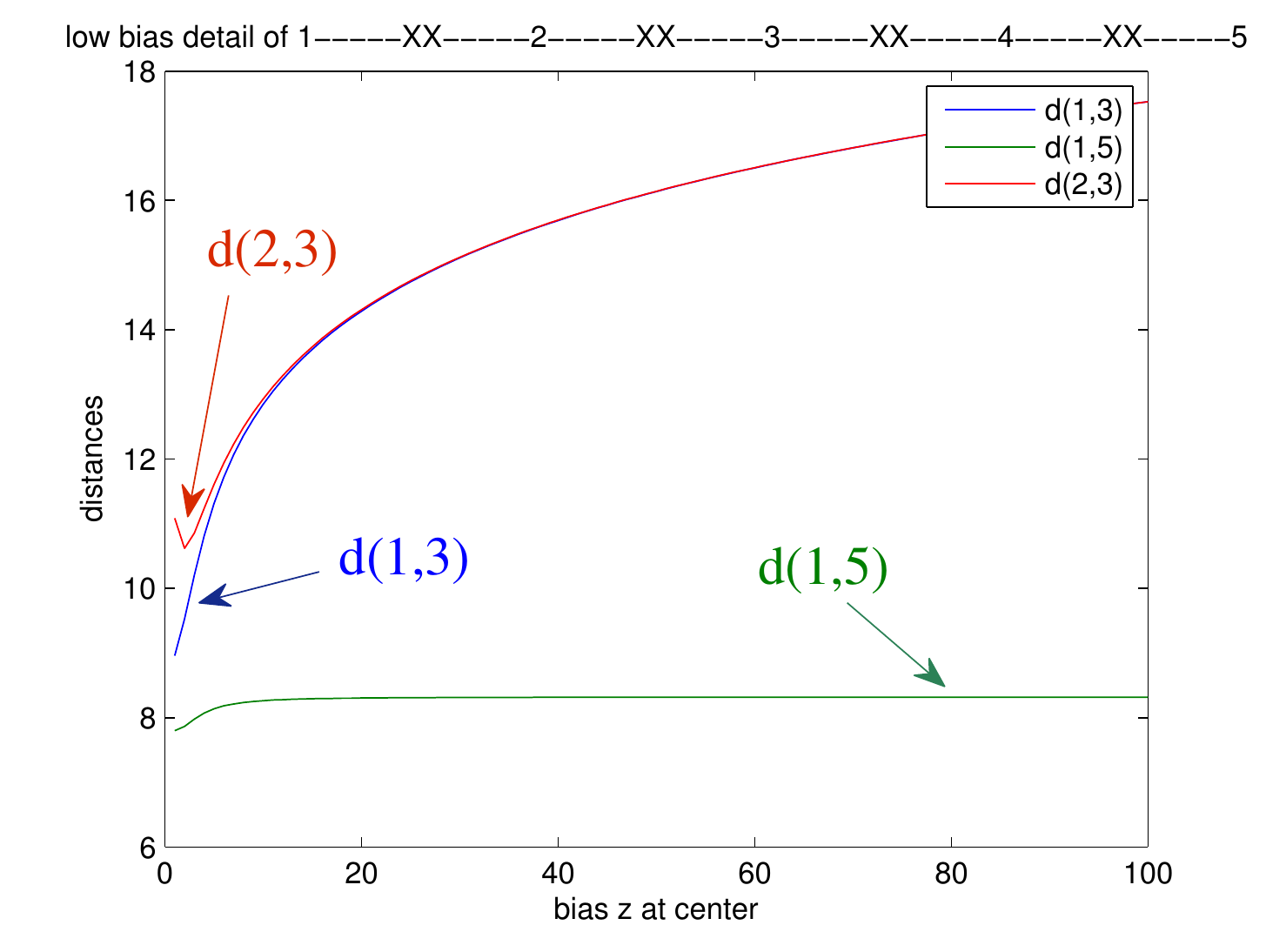}}}}
\vskip-.5cm
\caption{Simple 3-spin (left) and 5-spin (right) chain examples with bias $\zeta$ at center 
showing logarithmic behavior of the distance from the center to the outer spins  
and vanishing (left) and bounded (right) distance between the outer spins.}
\label{f:chain123}
\end{figure}
%**********************************************

There are several important observations to be made from this simple example: 
\begin{itemize}
\item The logarithmic behavior of the distance between the anti-core $2$
and the outer spin, $d(1,2)=\Theta(\log(\zeta))$, {\it is confirmed
analytically} from the symbolic expression of the eigenvectors. The same
applies to $d(1,3)=0$.

\item Because $d(1,3)=0$ and $d(2,1),d(2,3)\to \infty$ as $\zeta \to
\infty$, the geodesic triangle $123$ degenerates to the ray
$[2,1]=[2,3]$.  (A {\it ray} is an isometric embedding of $[0,\infty)$
to a metric space~\cite[III.H.3]{BridsonHaefliger1999}, intuitively
meaning that a ray starts at a finite point and extends to infinity
along a length minimizing path.)  As such, because $\triangle 123$ is
``flat,'' it is a Gromov $\delta$-slim triangle.  (A triangle is
$\delta$-slim~\cite[III.H.1]{BridsonHaefliger1999} if any edge is
contained in the union of the $\delta$-neighborhoods of the other two
edges.)  Thus the $123$ chain is a, albeit trivial, Gromov hyperbolic
space.  (A metric space is Gromov
hyperbolic~\cite[III.H.1]{BridsonHaefliger1999} if there exists a
$\delta < \infty$ such that all of its geodesic triangles are
$\delta$-slim.)

\item The rays $[2,1]$ and $[2,3]$ are going to infinity while keeping
their Hausdorff distance finite, in fact vanishing.  So they converge to
the same point on the Gromov boundary.  (The Gromov
boundary~\cite[III.H.3]{BridsonHaefliger1999} is the equivalence class
of rays keeping their distance finite.)  Thus the $123$ chain has its
Gromov boundary reduced to a singleton.

\item The preceding fact (Gromov boundary reduced to a singleton) is the
major topological discrepancy between classical and quantum
networks. Most classical networks have at least two points in their
Gromov boundary, creating a core~\cite{baryshnikov_tucci} as opposed to
the anti-core of quantum communications.
\end{itemize}

The preceding 3-spin example has been extended to the same 5-spin case
(Fig.~\ref{f:chain123}, right) with all results {\it proved} by symbolic
manipulations, which unfortunately become prohibitively long to be
included here.  The results are the same, except that the distance
between the outer spins remains finite (rather than vanishing), but this
suffices to come to the conclusion that the chain has only one point in
its Gromov boundary. We conjecture that this is a general feature.

\subsection{Infinite Chains}

A proof of the infinite diameter property of the engineered chain with
bias $\zeta$ at the center $\omega$ can be developed under infinite
number of spins hypothesis, $N \to \infty$.  Indeed, in this case, $H_1$
becomes the doubly infinite Toeplitz (also referred to as Laurent or
multiplication) operator $T_s$ with symbol
$s(\exp(i\theta))=\exp(i\theta)+\exp(-i\theta)$, and $H_1^{(\zeta)}$ is
a compact perturbation $\zeta E$ of $T_s$, where $E=|e\rangle\langle e|$
with $e=(\ldots0,0,1,0,0\ldots)'$ the unit basis vector of the Hilbert
space $\ell^2(-\infty,+\infty)$ of square summable doubly infinite
sequences. By a well known perturbation theory result, the spectrum of
$H_1^{(\zeta)}:\ell^2(-\infty,+\infty) \to \ell^2(-\infty,+\infty)$
consists of the interval $s([0,2\pi))=[-2,+2]$ plus another eigenvalue
of finite multiplicity that converges asymptotically to $\zeta$. The
eigen-equation $H_1^{(\zeta)}v_\lambda=\lambda v_\lambda$ becomes
\[  (T_s-\lambda I)v_\lambda=- \zeta Ev_\lambda=-\zeta v_{\lambda,\omega} e,\] 
where $v_{\lambda,\omega}$ denotes the central component of the
eigenvector $v_\lambda$.  In the Fourier or $z$-domain, $T_s-\lambda I$
is just the multiplication by $(z+z^{-1}-\lambda)$ operator. Therefore,
the above can be resolved as
\[\hat{v}_\lambda = \frac{-\zeta v_{\lambda,\omega}}{z+z^{-1}-\lambda}, \]
where $\hat{v}_\lambda$ denotes the Fourier or $z$-transform of
$v_\lambda$.  Taking $\lambda \in (-2,2)$, it is easily verified that
the Laurent expansion of the above converges on the unit circle; hence
the inverse Fourier or $z$-transform is in $\ell^2(-\infty,+\infty)$.
This provides an example of the rather unusual circumstance under which
a {\it continuous} spectrum (here $[-2,+2]$ of $T_s$) is converted in to
{\it pure point} spectrum (here $(-2,+2)$ of $T_s+\zeta E$) by a compact
perturbation.

The formula for the probability becomes
\[ p_{\mathrm{max}}(i,j)
  =\int_{-2}^{+2} |\langle v_\lambda | i \rangle \langle j | v_\lambda
  \rangle | d \lambda + |\langle v_\zeta | i \rangle   \langle j | v_\zeta \rangle |,
\]
where $v_\zeta$ is the eigenvector corresponding to the asymptotic
eigenvalue $\zeta$. Here, we are specifically interested in the case
$p_{\mathrm{max}}(\omega,\infty)$ of the probability of transition from
the center to infinity. Take $\lambda \in (-2,+2)$. We need to evaluate
$\langle v_\lambda | \omega\rangle$ and $\langle v_\lambda |
\infty\rangle$. Recall that Parseval's theorem
$\sum_{m=-\infty}^{+\infty}a_mb_m=\frac{1}{2\pi i} \oint \hat{a}(z)
\hat{b}(z^{-1}) \frac{dz}{z}$ allows us to compute an inner product by
residue calculation. Using the recipe yields
\begin{align*}
\langle v_\lambda | \omega \rangle
&= \frac{1}{2\pi \imath}\oint 1 \frac{-\zeta v_{\lambda,\omega}}{z+z^{-1}-\lambda} \frac{dz}{z}
= \frac{1}{2\pi \imath}\oint 1 \frac{-\zeta v_{\lambda,\omega}}{z^2-\lambda z +1} dz\\
&=\mathrm{Residue}\left( \frac{-\zeta v_{\lambda,\omega}}{z^2-\lambda z +1}\right)_{p,\bar{p}},
\end{align*} 
where $p,\bar{p}$, with $|p|<1, |\bar{p}|<1$, are the poles of the
integrand, that is, the zeros of $z^2-\lambda z + 1$. This yields
\begin{equation*}  
  \langle v_\lambda | \omega \rangle 
  = \frac{-\zeta v_{\lambda,\omega}}{p-\bar{p}} + \frac{-\zeta v_{\lambda,\omega}}{\bar{p}-p} 
  =0.  
\end{equation*}
Next, we look at the term $\langle v_\lambda |\infty\rangle$ as the
limit of $\langle v_\lambda |M\rangle$ as $M \to \infty$. We have
\begin{align*}
\langle v_\lambda |M\rangle 
&= \frac{1}{2\pi \imath}\oint \frac{-\zeta v_{\lambda,\omega}}{z+z^{-1}-\lambda}z^M \frac{dz}{z}
 = \mathrm{Residue}\left( \frac{-\zeta v_{\lambda,\omega} z^M}{z^2-\lambda z +1} \right)_{p,\bar{p}}\\
&= \frac{-\zeta v_{\lambda,\omega}p^M}{p-\bar{p}} + \frac{-\zeta v_{\lambda,\omega}\bar{p}^M}{\bar{p}-p}.
\end{align*}
Since $|p|<1$, the limit of the above as $M\to \infty$ vanishes. Last,
we look at the isolated eigenvalue case, $|\langle v_k | \omega \rangle
\langle M | v_k \rangle |$.  Since $v_k \approx e$, and the excited
state $|M\rangle$ is a basis vector orthogonal to $e$, we have $\langle
M | v_k \rangle=0$. Hence the probability
$p_{\mathrm{max}}(\infty,\omega)=0$ and the distance
$d_\zeta(\omega,\infty)$ is infinity.

\section{Discussion and Conclusions}

The early numerical observation~\cite{first_with_sophie} that
homogeneous odd length $N$-spin chains have an ``anti-gravity'' center
has been {\it analytically} confirmed in the $N \to \infty$ limit by
developing closed-form formulas for the asymptotic maximum excitation
transfer probability, and by showing that the anti-core has the lowest
probability of being excited or of transmitting its excitation. As shown
in Sec.~\ref{s:infinite_chains}, the phenomena exhibited in
Fig.~\ref{f:dramatic_anti_core} can be accurately explained by the
$N=\infty$ asymptotic formula.

The existence of an anti-core at the center of a linear array of spins
shows that excitations in a spin network do not propagate as they would
in a classical network.  In a classical linear network any excitation in
one half of the chain must transit through the center to reach the other
half, and the center would thus be expected to be a congestion core.  In
quantum networks, however, excitations can be transferred from one end
of a chain to the other without passing through the center due to the
intrinsic entanglement present in the eigenstates of the system.  When a
single excitation is created in one location what is really created is a
wavepacket, which is a superposition of many eigenstates of the system
Hamiltonian that subsequently evolve and interfere.  High probability of
transmission requires constructive interference at any particular node
at some time, and the existence of an anti-core shows that, surprisingly,
there is least constructive interference for the center of the chain.

As we have also shown, it is possible to engineer chains such that the
diameter of the chain goes to infinity even if the physical number of
spins is finite.  The simplest way to achieve this is by applying a bias
to the central spin in an odd-length chain.  By increasing the bias we
can increase the diameter even for a finite chain and achieve infinite
diameter in the limit of infinite bias.  This can be explained in terms
of the bias moving the central spin further and further away from the
other spins and therefore effectively decoupling the chains.  However,
for any finite bias, no matter how large, an excitation in one half of
the chain can tunnel through the obstruction in the center given
sufficient time, allowing almost perfect excitation transfer between the
end spins of the engineered chain.

Such finite length, infinite diameter chains lend themselves to a coarse
Gromov analysis, as Section~\ref{s:engineered} shows.  The specific
feature, demonstrated on chains of limited length but conjectured to
hold for longer chains, is a Gromov boundary reduced to a singleton.
This strongly contrasts with the classical network paradigm of a Gromov
boundary with at least two points, creating the congestion
core~\cite{baryshnikov_tucci}.  Although there is early indication that
the difference in cardinalities of the Gromov boundaries might be at
least part of the explanation of the core versus anti-core discrepancy,
more analysis is needed to prove a general fact and is left for further
research.

\appendix

\section{Proof of Theorem~\ref{t:semi_infinite_chains} (semi-infinite chain)}
\label{s:semi_infinite_chains}

We proceed from
\begin{equation} 
\label{e:pmax_explicit_in_N}
\sqrt{p_{\mathrm{max}}^{[1:N]}(i,j)}=\frac{2}{N+1} \sum_{k=1}^{2n+1}
\left| \sin \frac{\pi k i}{2(n+1)} \sin \frac{\pi k j}{2(n+1)} \right|, 
\end{equation}
where $N=2n+1$ is the (odd) number of spins and $i$ and $j$ are the
positions of the two spins {\it relative to the left-most spin (1).}
Since the number of spins will be taken to infinity, we make the
dependency on such number explicit.

\subsection{Asymptotic maximum transfer probability}

Defining $x_k'=k/(2(n+1))$ for $k=0,\ldots,2n+1$, the right-hand side
of~\eqref{e:pmax_explicit_in_N} becomes
\begin{equation*}
  \sqrt{p_{\mathrm{max}}^{[1:N]}(i,j)} = 
  2\sum_{k=1}^{2n+1} \left| \sin \pi x_k' i\right| \left| \sin \pi x_k' j \right| (x_{k}'-x_{k-1}').
\end{equation*}
Taking the limit $n \to \infty$, the above becomes
\begin{align*}
\sqrt{p_{\mathrm{max}}^{\rightarrow}(i,j)}
&:=\lim_{n \to \infty} \sqrt{p_{\mathrm{max}}^{[1:N]}(i,j)} \\ 
&= 2 \int_0^1 |\sin  \pi i x'||\sin \pi j x'| \, dx'  
 = 4 \int_0^{1/2} |\sin 2 \pi i x||\sin 2 \pi j x|\, dx.
\end{align*}
Since $|\sin 2 \pi i x| = \left|\sin \left(2 \pi
i\left(x+\frac{1}{2}\right)\right)\right|$, the above becomes
\begin{align*}
\sqrt{p_{\mathrm{max}}^{\rightarrow}(i,j)}
=& 2 \int_0^{1/2} \!\!\! |\sin 2 \pi i x||\sin 2 \pi j x|dx\\
  &~~~~~+2 \int_0^{1/2} \!\!\! \left|   \sin 2 \pi i \left( \tfrac{1}{2}+x \right)   \right|
        \left| \sin 2 \pi j \left(\tfrac{1}{2}+x \right)\right|dx\\
=& 2 \int_0^1 |\sin 2 \pi i x||\sin 2 \pi j x|dx.
\end{align*}
Next, observe that $|\sin 2 \pi i x|=\sin (2 \pi i x) s_i(x)$ where
$s_i(x)$ is a periodic square wave with fundamental $\sin 2 \pi i x$ and
Fourier decomposition
\[ s_i(x)=\frac{4}{\pi} \sum_{p=1,3,\ldots} \frac{1}{p}\sin 2 \pi i p x. \]
Therefore, the absolute values in the integral representation of
$\sqrt{p_{\mathrm{max}}(i,j)}$ can be removed as follows:
\begin{align*}
 \sqrt{p_{\mathrm{max}}^{\rightarrow}(i,j)}
 =\frac{32}{\pi^2}\sum_{p,q=1,3,\ldots}\frac{1}{pq}\int_0^1 
(\sin 2 \pi i x \sin 2 \pi i p x)( \sin 2 \pi j x \sin 2 \pi j q x) dx. 
\end{align*}
Next, utilizing several well-known trigonometric identities, we get, successively, 
\begin{align*}
\lefteqn{\sqrt{p_{\mathrm{max}}^{\rightarrow}(i,j)}}\\
=&\frac{8}{\pi^2} \sum_{p,q}\frac{1}{pq}\int_0^1 
  (\cos 2 \pi i (p-1)x-\cos 2 \pi i (p+1)x)\\
&~~~~~~~~~~~~~~~~~~~~~~~\cdot(\cos 2 \pi j (q-1)x - \cos 2 \pi j (q+1)x) \, dx\\
=& \frac{4}{\pi^2} \sum_{p,q}\frac{1}{pq}\int_0^1 (
  (\cos 2 \pi (i(p-1)-j(q-1))x+\cos 2 \pi(i(p-1)+j(q-1))x) \\
~&\qquad \qquad -(\cos 2 \pi (i(p-1)-j(q+1)x)+\cos 2 \pi (i(p-1)+j(q+1))x)\\
~&\qquad \qquad  -(\cos 2 \pi (i(p+1)-j(q-1))x+\cos 2\pi ( i (p+1)+j(q-1))x)\\
~&\qquad \qquad +(\cos 2 \pi (i(p+1)-j(q+1))x+\cos 2 \pi (i(p+1)+j(q+1))x)\, )dx\\
=&\frac{4}{\pi^2} \sum_{p,q}\frac{1}{pq}(\mathbbm{1}_{i(p-1)-j(q-1)=0}+\mathbbm{1}_{i(p-1)+j(q-1)=0}\\
~& \qquad\qquad -\mathbbm{1}_{i(p-1)-j(q+1)=0}-\mathbbm{1}_{i(p-1)+j(q+1)=0}\\
~& \qquad\qquad -\mathbbm{1}_{i(p+1)-j(q-1)=0}-\mathbbm{1}_{i(p+1)+j(q-1)=0}\\
~& \qquad\qquad +\mathbbm{1}_{i(p+1)-j(q+1)=0}+\mathbbm{1}_{i(p+1)+j(q+1)=0}\, ),
\end{align*}
where $\mathbbm{1}_{L}$ takes the value $1$ if the logical statement $L$
is true and $0$ otherwise.

Observe that, since $i,p,j,q \geq1$,
\[ \mathbbm{1}_{i(p-1)+j(q+1)=0}=0, \:\mathbbm{1}_{i(p+1)+j(q-1)=0}=0, \:\mathbbm{1}_{i(p+1)+j(q+1)=0}=0\]
and that $\mathbbm{1}_{i(p-1)+j(q-1)=0}$ takes the value $1$ only for $p=q=1$. 
Thus, in the above, the sum over $p,q$ of the right-hand side terms amounts to $1$. 

Next, we look at the left-hand side terms of the sum over $p,q$.  For
such a statement as $i(p-1)-j(q-1)=0$ to be true, we need $p-1=m{\bf j}$
and $q-1=m{\bf i}$ for some $m \in \mathbb{N}$, where
$\mathbf{j}=j/\mathrm{gcd}(i,j)$ and ${\bf
i}=i/\mathrm{gcd}(i,j)$.  This yields $p=m{\mathbf j}+1$ and $q=m{\bf
i}+1$.  Observe that $m=0$ is an admissible value, since this yields
$p=1$, $q=1$ and hence $i(p-1)-j(q-1)=0$.  Recapitulating and following
up with the same argument on the other logical statements, we find
\begin{align*}
i(p-1)-j(q-1)=0 & \Leftrightarrow \left\{ \begin{array}{c}
                                                                  p=m{\bf j}+1,\\
                                                                  q=m{\bf i}+1,
                                                                  \end{array}\right.\\
i(p-1)-j(q+1)=0 & \Leftrightarrow \left\{ \begin{array}{c}
                                                                  p=m{\bf j}+1,\\
                                                                  q=m{\bf i}-1,
                                                                  \end{array}\right.\\
i(p+1)-j(q-1)=0 & \Leftrightarrow \left\{ \begin{array}{c}
                                                                  p=m{\bf j}-1,\\
                                                                  q=m{\bf i}+1,
                                                                  \end{array}\right.\\
i(p+1)-j(q+1)=0 & \Leftrightarrow \left\{ \begin{array}{c}
                                                                  p=m{\bf j}-1,\\
                                                                  q=m{\bf i}-1.
                                                                  \end{array}\right.
\end{align*}
Note that, since $p,q \geq 1$, the solution $m=0$ is not admissible for
the second, third, and forth cases, since this would entail either $p$
or $q$ or both of them to equal $-1$.  In addition, $p$ and $q$ must be
restricted to be odd, that is, both $m{\bf i}$ and $m{\bf j}$ must be
even.  Since ${\bf i}$ and ${\bf j}$ are relatively prime, they cannot
be both even; thus $m$ must be even.  The solution $m=0$ is still
acceptable for the first case, since it makes both $p$ and $q$ odd.

To summarize, the sum over $p,q=1,3,\ldots$ reduces to $2$ plus a sum
over $m \in \mathbb{N}^*$, subject to the restrictions that $m{\bf i}-1
\ne 0$ and $m {\bf j}-1 \ne 0$.  Changing the sum over $p,q$ to a sum
over $m$, it follows that
\begin{align}
\label{e:series_representation}
\sqrt{p_{\mathrm{max}}^{\rightarrow}(i,j)}
=& \frac{4}{\pi^2} \left[ 2+ \sum_{m\in M} 
   \frac{1}{(m{\bf j}+1)(m{\bf i}+1)}-\frac{1}{(m{\bf j}+1)(m{\bf i}-1)} \right. \nonumber\\
 & \;\qquad\qquad\left. -\frac{1}{(m{\bf j}-1)(m{\bf i}+1)}+\frac{1}{(m{\bf j}-1)(m{\bf i}-1))} \right]
  \nonumber\\
=& \frac{4}{\pi^2}\left( 2+\sum_{m\in M} \frac{4}{(m^2{\bf j}^2-1)(m^2{\bf i}^2-1)}   \right),
\end{align}
where $M = \{ m \in \mathbb{N}^*: m \mbox{ is even}, m^2{\bf i}^2-1\ne 0, m^2{\bf j}^2-1 \ne 0\}$. 
This proves the infinite series representation of Theorem~\ref{t:semi_infinite_chains}. $\blacksquare$

\subsection{Closed form of asymptotic transfer probability}

Next, we express the infinite series of
Theorem~\ref{t:semi_infinite_chains} in terms of elementary
functions. First, consider the partial fraction decomposition
\begin{equation*}
 \frac{1}{(m^2{\bf j}^2-1)(m^2{\bf i}^2-1)}   
  =\frac{A(\bf{i},\bf{j})}{(m^2{\bf j}^2-1)}+\frac{A(\bf{j},\bf{i})}{(m^2{\bf i}^2-1)},
 \qquad
 A(\bf{i},\bf{j}) := \frac{\bf{i}^2}{\bf{j}^2-i^2}. 
\end{equation*}
Next, observe the following lemma:
\begin{lemma}
\label{l:cotangent}
\[ 
   \sum_{m=2,4,6,...} \frac{1}{(m^2{\bf i}^2-1)} =
   \frac{1}{2} \left( 1-\frac{\pi}{2\bf{i}} \cot \frac{\pi}{2\bf{i}}\right). 
\]
\end{lemma}
\begin{proof}
In the known expression for the cotangent,
\[ 
  \pi \cot(\pi z)=\frac{1}{z}+2z \sum_{n=1}^{\infty}\frac{1}{z^2-n^2}, 
\]
set $z=1/2\bf{i}$. This yields
\[ 
  \sum_{n=1}^\infty \frac{1}{(2n\bf{i})^2-1}= 
  \frac{1}{2}\left( 1- \left( \frac{\pi}{2\bf{i}}  \right) \cot \left(\frac{\pi}{2\bf{i}}  \right)\right). 
\]
Setting $m=2n$ yields the result. 
\end{proof}

Putting everything together using the lemma yields
\[ 
  \sqrt{p_{\mathrm{max}}^{\rightarrow}(i,j)}
  =\frac{8}{\pi^2} \left( \frac{\bf{i}^2}{\bf{i}^2-\bf{j}^2}
                   \left(\frac{\pi}{2\bf{i}}\right)\cot\left(\frac{\pi}{2\bf{i}} \right)
   -\frac{\bf{j}^2}{\bf{i}^2-\bf{j}^2} \left(\frac{\pi}{2\bf{j}} \right)\cot\left(\frac{\pi}{2\bf{j}} 
  \right) \right) 
\]
and Theorem~\ref{t:semi_infinite_chains} is proved. $\blacksquare$

\section{Proofs of Theorem~\ref{t:doubly_infinite_chains_same} (doubly-infinite chain)}
\label{s:doubly_infinite_chains}

We proceed from
\[ \sqrt{p_{\mathrm{max}}^{[1:N]}(i,j)}=\frac{2}{N+1} \sum_{k=1}^{2n+1}
\left| \sin \frac{\pi k i}{2(n+1)} \sin \frac{\pi k j}{2(n+1)} \right|, \]
where $N=2n+1$ is the (odd) number of spins and $i$ and $j$ are the positions of the two spins 
{\it relative to the central spin (n+1).} 
Since  the number of spins will be taken to infinity, we make the dependency on such number explicit. 

\subsection{Referencing max.\ transfer probability to anti-core}

The first operation is to do the change of variable $k'=k-(n+1)$ and
convert the sum as $k$ goes from $1$ to $2n+1$ to a sum where $k'$ goes
from $-n$ to $+n$. After some manipulation, the following is found:
\begin{equation*}
  \sqrt{p_{\mathrm{max}}^{[1:N]}(i,j)}
  =\frac{2}{N+1} \sum_{k'=-n}^{+n}  
    \left|f\left( \frac{\pi k' i}{2(n+1)}\right)g\left(\tfrac{\pi k' j}{2(n+1)}\right)\right|, 
\end{equation*}
where $f,g$ are given in Table~\ref{t:f_and_g}. 

\begin{table}[h]
\caption{The functions $f$ and $g$.}
\begin{center}
\begin{tabular}{||c|c|c||}\hline\hline
               & $f(\cdot)$ & $g(\cdot)$ \\\hline\hline
$i,j$ even & $\sin(\cdot)$ & $\sin(\cdot)$ \\\hline
$i,j$ odd  & $\cos(\cdot)$ & $\cos(\cdot)$ \\\hline
$i$ even, $j$ odd & $\sin(\cdot)$ & $\cos(\cdot)$ \\\hline
$i$ odd, $j$ even & $\cos(\cdot)$ & $\sin(\cdot)$ \\\hline\hline
\end{tabular} 
\end{center}
\label{t:f_and_g}
\end{table}

The next step is the change of variables $i'=i-(n+1)$ and $j'=j-(n+1)$.
With this change of variables, the position of the spins are relative to
the anti-core, $n+1$.  This change of variables leads to the following:
\begin{align*}
\sqrt{p^{[-n:+n]}_{\mathrm{max}}(i',j')}
&:= \sqrt{p_{\mathrm{max}}^{[1:N]}(i'+(n+1),j'+(n+1))}\\
& = \frac{2}{N+1} \sum_{k'=-n,\mathrm{even}}^{+n}
   \left|f\left( \frac{\pi k' i'}{2(n+1)}\right)g\left( \frac{\pi k' j'}{2(n+1)}\right)\right|\\
& \quad +\frac{2}{N+1} 
   \sum_{k'=-n,\mathrm{odd}}^{+n}
   \left|\bar{f}\left( \frac{\pi k' i'}{2(n+1)}\right)\bar{g}\left( \frac{\pi k' j'}{2(n+1)}\right)\right|,
\end{align*}
where $f,g$ are still given by Table~\ref{t:f_and_g} and
$\bar{f}(\cdot)=\cos(\cdot),\sin(\cdot)$ whenever
$f(\cdot)=\sin(\cdot),\cos(\cdot)$, resp., with the same definition for
$\bar{g}$.  Since the most recent formula is in terms of $i',j'$, we
rewrite Table~\ref{t:f_and_g} in terms of $i',j'$ and $n$ in
Table~\ref{t:in_terms_of}.

\begin{table}[h]
\caption{The functions $f$ and $g$ in terms of $i',j'$ and $n$.}
\begin{center}
\begin{tabular}{||c|c|c||}\hline\hline
               & $f(\cdot)$ & $g(\cdot)$ \\\hline\hline
$i'+n,j'+n$ odd & $\sin(\cdot)$ & $\sin(\cdot)$ \\\hline
$i'+n,j'+n$ even  & $\cos(\cdot)$ & $\cos(\cdot)$ \\\hline
$i'+n$ odd, $j'+n$ even & $\sin(\cdot)$ & $\cos(\cdot)$ \\\hline
$i'+n$ even, $j'+n$ odd & $\cos(\cdot)$ & $\sin(\cdot)$ \\\hline\hline
\end{tabular} 
\end{center}
\label{t:in_terms_of}
\end{table}

\subsection{Towards asymptotic maximum probability}

In anticipation of letting $n \to \infty$, define
$x_{k'}:=\frac{k'}{4(n+1)}$ and the preceding sums can be rewritten as
\begin{align*}
\sqrt{p^{[-n:+n]}_{\mathrm{max}}(i',j')}
&= 2\sum_{k'=-n,\mathrm{even}}^{+n}
   \left|f\left( 2 \pi x_{k'}i'\right) g\left( 2\pi x_{k'} j'\right)\right|\left(x_{k'+2}-x_{k'} \right)\\
&\quad + 2 
  \sum_{k'=-n,\mathrm{odd}}^{+n}
  \left|\bar{f}\left( 2 \pi x_{k'}i' \right)
        \bar{g}\left( 2\pi x_{k'} j' \right)\right| \left(x_{k'+2}-x_{k'} \right). 
\end{align*}
Letting $n \to \infty$ yields
\begin{align*}
\sqrt{p^{\leftrightarrow}_{\mathrm{max}}(i',j')}
&:= \lim_{n\to \infty}\sqrt{p^{[-n:+n]}_{\mathrm{max}}(i',j')}\\
&=       2\int_{-1/4}^{+1/4} |f(2\pi xi')g(2\pi x j')| dx  
       + 2\int_{-1/4}^{+1/4}|\bar{f}(2\pi xi')\bar{g}(2\pi x j')| dx.
\end{align*} 

In order to make the integrations more straightforward and to follow a
procedure that parallels the one of
Appendix~\ref{s:semi_infinite_chains}, it is convenient to change the
integration limits by making use of the periodicity of the integrands as
functions of $x$.  Observe that both $fg$ and $\bar{f}\bar{g}$ have
decompositions in terms of sines or cosines of arguments $2\pi x (i' \pm
j')$. Write the generic term as $\left\{\begin{array}{c} \sin\\ \cos
\end{array}\right\}(2\pi x (i' \pm j'))$. 
If $i'\pm j'$ is even, observe that 
$\left\{\begin{array}{c}
\sin\\
\cos
\end{array}\right\}(2\pi (x+1/2) (i' \pm j'))=
\left\{\begin{array}{c}
\sin\\
\cos
\end{array}\right\}(2\pi x (i' \pm j'))$. 
If $i'\pm j'$ is odd, 
$\left\{\begin{array}{c}
\sin\\
\cos
\end{array}\right\}(2\pi (x+1/2) (i' \pm j'))=
-\left\{\begin{array}{c}
\sin\\
\cos
\end{array}\right\}(2\pi x (i' \pm j'))$. 
In either case, $|fg|$ and $|\bar{f}\bar{g}|$ have period $1/2$. 
With this property, the previous integrals can be rewritten as 
\begin{equation}
\label{e:general}
\sqrt{p^{\leftrightarrow}_{\mathrm{max}}(i',j')}
= \int_{-1/2}^{+1/2} |f(2\pi xi')g(2\pi x j')| dx 
 +\int_{-1/2}^{+1/2}|\bar{f}(2\pi xi')\bar{g}(2\pi x j')| dx.
\end{equation} 
Observe that $p_{\mathrm{max}}(i',j') \leq 1$, as easily seen from a Cauchy-Schwarz argument. 
Also observe that $p_{\mathrm{max}}(i',i')=1$; indeed, if $i'=j'$, 
Table~\ref{t:in_terms_of} reveals that the integrands are of the form 
$|\cos(2\pi i'x)\cos(2\pi i'x)|$ or $|\sin(2\pi i'x)\sin(2\pi i'x)|$, 
from which the assertion is trivial. 

The next step is to make $|f|$, $|g|$, $|\bar{f}|$, $|\bar{g}|$ more
manageable by expressing them as $f(2\pi xi')s_{i'}(x)$, $g(2\pi x
j')s_{j'}(x)$, $\bar{f}(2\pi x i')s_{i'}(x)$, $\bar{g}(2\pi x
j')s_{j'}(x)$ if $f$, $g$, $\bar{f}$, $\bar{g}$ are sines and by $f(2\pi
xi')c_{i'}(x)$, $g(2\pi x j')c_{j'}(x)$, $\bar{f}(2\pi x i')c_{i'}(x)$,
$\bar{g}(2\pi x j')c_{j'}(x)$ if they are cosines. In the preceding,
$s_{i'}(x)$ and $c_{i'}(x)$ are odd and even, resp., square waves of
unit amplitude and of period $1$, with Fourier decompositions
\begin{align*}
s_{i'}(x)&=\frac{4}{\pi} \sum_{p=1,3,...} \frac{1}{p} \sin (2 \pi i'px),\\
c_{i'}(x)&=\frac{4}{\pi} \sum_{p=1,3,...} \frac{(-1)^{\frac{p-1}{2}}}{p} \cos (2 \pi i'px).
\end{align*}

At this stage, it is necessary to be more specific as to what $f$, $g$, $\bar{f}$, $\bar{g}$ are.

\subsection{Consistency of asymptotic max.\ transfer probability}
\label{s:consistency}

\subsubsection{$i'$ and $j'$ even}

If $i'$ and $j'$ are even, and if we let $n \to \infty$ along the even
number subsequence of $\mathbb{N}$, we need to take
$f(\cdot)=\cos(\cdot)$ and $g(\cdot)=\cos(\cdot)$, as seen from
Table~\ref{t:in_terms_of}, together with $\bar{f}(\cdot)=\sin(\cdot)$
and $\bar{g}(\cdot)=\sin(\cdot)$.  If on the other hand, we let $n \to
\infty$ along the odd number subsequence of $\mathbb{N}$, we need to
take $f(\cdot)=\sin(\cdot)$ and $g=\sin(\cdot)$, together with
$\bar{f}(\cdot)=\cos(\cdot)$ and $\bar{g}(\cdot)=\cos(\cdot)$ .
However, because of the symmetry of formula~\eqref{e:general}, both
subsequences yield the same result:
\begin{align*}
\lefteqn{\sqrt{p^{\leftrightarrow}_{\mathrm{max}}(i',j')}}\\
&= \int_{-1/2}^{+1/2}|\cos(2\pi xi') \cos(2\pi x j')| dx
  +\int_{-1/2}^{+1/2}|\sin(2\pi xi') \sin(2\pi x j')| dx\\
&= \int_{-1/2}^{+1/2}\cos(2\pi xi') c_{i'}(x)\cos(2\pi x j') c_{j'}(x) dx    \\
  &~~~~~~~~~~~~~~~~~~~~~~~~~~~~+\int_{-1/2}^{+1/2}\sin(2\pi xi') s_{i'}(x)\sin(2\pi x j') s_{j'}(x) dx.
\end{align*}

\subsubsection{$i'$ and $j'$ odd}

The argument is the same as before and the preceding formula still holds.

\subsubsection{$i'$ odd and $j'$ even}

From Table~\ref{t:in_terms_of}, we could let $n \to \infty$ along the
even number subsequence of $\mathbb{N}$, in which case we need to take
$f(\cdot)=\sin(\cdot)$ and $g(\cdot)=\cos(\cdot)$.  If we let $n \to
\infty$ along the odd number subsequence of $\mathbb{N}$, we need to
take $f(\cdot)=\cos(\cdot)$ and $g(\cdot)=\sin(\cdot)$. In either case,
because of the symmetry of~\eqref{e:general}, the result is the same and
is given by
\begin{align*}
\lefteqn{\sqrt{p^{\leftrightarrow}_{\mathrm{max}}(i',j')}}\\
&=\int_{-1/2}^{+1/2} |\sin(2\pi xi') \cos (2\pi x j')| dx
 +\int_{-1/2}^{+1/2} |\cos(2\pi xi') \sin (2\pi x j')| dx\\
&=\int_{-1/2}^{+1/2}  \sin(2\pi xi') s_{i'}(x)\cos (2\pi x j') c_{j'}(x) dx\\
 &~~~~~~~~~~~~~~~~~~~~~~~~+\int_{-1/2}^{+1/2}  \cos(2\pi xi') c_{i'}(x)\sin (2\pi x j') s_{j'}(x) dx.
\end{align*}

\subsubsection{$i'$ even and $j'$ odd}

The formula of the preceding section remains valid. To prove it, just
interchange the role of $i'$ and $j'$.

\subsection{Towards integration by quadrature}

From the above, it follows that all cases share a few quadrature
integrals:
\begin{align*}
\lefteqn{\int_{-1/2}^{+1/2}\sin(2\pi xi') s_{i'}(x)\sin(2\pi x j') s_{j'}(x)dx}\\
&=\frac{16}{\pi^2} \sum_{p,q=1,3,...} \frac{1}{pq}
   \int_{-1/2}^{1/2}\sin(2\pi i'x)\sin(2\pi i'px)\sin(2\pi j'x)\sin(2\pi j'qx)dx\\
&=\frac{4}{\pi^2} \sum_{p,q}\frac{1}{pq} 
   \int_{-1/2}^{1/2}(\cos(2\pi i'(p-1)x)-\cos(2\pi i' (p+1)x))\\
  &~~~~~~~~~~~~~~~~~~~~~~~~~~~~ \cdot    (\cos(2\pi j'(q-1)x)-\cos(2\pi j' (q+1)x))dx\\
&=\frac{2}{\pi^2} \sum_{p,q} \frac{1}{pq}\int_{-1/2}^{1/2}
                 (\sum \cos(2\pi(i'(p-1)\pm j'(q-1))x)\\
&~~~~~~~~~~~~~~~~~~~~~~~~~~~~~~~~~~~~-\sum \cos(2\pi(i'(p-1)\pm j'(q+1))x)\\
& \quad           -\sum \cos(2\pi(i'(p+1)\pm j'(q-1))x)+\sum \cos(2\pi(i'(p+1) \pm j'(q+1))x))dx.
\end{align*}
The right-hand side of the last equality involves expressions like
$\cos(2\pi(a+b)x)+\cos(2\pi(a-b)x)$. To simplify the notation, we wrote
such expressions as $\sum \cos(2\pi(a\pm b)x)$.

\begin{align*}
\lefteqn{(-1)^{\frac{p+q}{2}-1}\int_{-1/2}^{+1/2}\cos(2\pi xi') c_{i'}(x)\cos(2\pi x j') c_{j'}(x)dx}\\
&= \frac{16}{\pi^2} \sum_{p,q=1,3,...}\frac{1}{pq} \int_{-1/2}^{1/2}\cos(2\pi i'x)\cos(2\pi i'px)\cos(2\pi j'x)\cos(2\pi j'qx)dx\\
&=\frac{4}{\pi^2} \sum_{p,q}\frac{1}{pq} \int_{-1/2}^{1/2}(\cos(2\pi i'(p+1)x)+\cos(2\pi i' (p-1)x))\\
               &~~~~~~~~~~~~~~~~~~~~~~~~~~\cdot(\cos(2\pi j'(q+1)x)+\cos(2\pi j' (q-1)x))dx\\
&=\frac{2}{\pi^2} \sum_{p,q} \frac{1}{pq}\int_{-1/2}^{1/2}
               (\sum \cos(2\pi(i'(p+1)\pm j'(q+1))x)\\
        &~~~~~~~~~~~~~~~~~~~~~~~~~~~~ +\sum \cos(2\pi(i'(p+1)\pm j'(q-1))x)\\
&\quad             +\sum \cos(2\pi(i'(p-1)\pm j'(q+1))x)
              +\sum \cos(2\pi(i'(p-1) \pm j'(q-1))x))dx;
\end{align*}

\begin{align*}
\lefteqn{(-1)^{\frac{q-1}{2}}\int_{-1/2}^{+1/2}\sin(2\pi xi') s_{i'}(x)\cos(2\pi x j') c_{j'}(x)dx}\\
&= \frac{16}{\pi^2} \sum_{p,q=1,3,...}\frac{1}{pq} \int_{-1/2}^{1/2}\sin(2\pi i'x)\sin(2\pi i'px)\cos(2\pi j'x)\cos(2\pi j'qx)dx\\
&=\frac{4}{\pi^2} \sum_{p,q} \frac{1}{pq}\int_{-1/2}^{1/2}(\cos(2\pi i'(p-1)x)-\cos(2\pi i' (p+1)x))\\
                 &~~~~~~~~~~~~~~~~~~~~~~~~\cdot (\cos(2\pi j'(q-1)x)+\cos(2\pi j' (q+1)x))dx\\
&=\frac{2}{\pi^2} \sum_{p,q} \frac{1}{pq}\int_{-1/2}^{1/2}
               (\sum \cos(2\pi(i'(p-1)\pm j'(q-1))x)\\
          &~~~~~~~~~~~~~~~~~~~~~~~~~~~~~~~~     +\sum \cos(2\pi(i'(p-1)\pm j'(q+1))x)\\
& \quad             -\sum \cos(2\pi(i'(p+1)\pm j'(q-1))x)
               -\sum \cos(2\pi(i'(p+1) \pm j'(q+1))x))dx;
\end{align*}

\begin{align*}
\lefteqn{(-1)^{\frac{p-1}{2}}\int_{-1/2}^{+1/2}\cos(2\pi xi') c_{i'}(x)\sin(2\pi x j') s_{j'}(x)dx}\\
&= \frac{16}{\pi^2} \sum_{p,q=1,3,...}\frac{1}{pq} \int_{-1/2}^{1/2}\cos(2\pi i'x)\cos(2\pi i'px)\sin(2\pi j'x)\sin(2\pi j'qx)dx\\
&=\frac{4}{\pi^2} \sum_{p,q} \frac{1}{pq}\int_{-1/2}^{1/2}(\cos(2\pi i'(p+1)x)+\cos(2\pi i' (p-1)x))\\
                 &~~~~~~~~~~~~~~~~~~~~~~~~~\cdot (\cos(2\pi j'(q-1)x)-\cos(2\pi j' (q+1)x))dx\\
&=\frac{2}{\pi^2} \sum_{p,q}\frac{1}{pq} \int_{-1/2}^{1/2}
               (\sum \cos(2\pi(i'(p+1)\pm j'(q-1))x)\\
    &~~~~~~~~~~~~~~~~~~~~~~~~~~~~~~~~~~          -\sum \cos(2\pi(i'(p+1)\pm j'(q+1))x)\\
& \quad +\sum \cos(2\pi(i'(p-1)\pm j'(q-1))x)
               -\sum \cos(2\pi(i'(p-1) \pm j'(q+1))x))dx.
\end{align*}

\subsection{Asymptotic max.~transfer probability around anti-core}

Here we proceed from the general formula~\eqref{e:general} for
$\sqrt{p_{\mathrm{max}}^{\leftrightarrow}(i',j')}$, utilize the
quadrature integrals of the preceding section, and derive, first, an
infinite series representation of the asymptotic maximum transfer
probability and, finally, a representation in terms of special
functions.  Since the general formula~\eqref{e:general} is in terms of
function $f$, $g$, $\bar{f}$, $\bar{g}$ that depend on whether $i'$ and
$j'$ are even or odd (see Section~\ref{s:consistency}), it is necessary
to examine each case in particular.  From Section~\ref{s:consistency},
it follows that the case where both $i'$ and $j'$ are even and the case
where both $i'$ and $j'$ are odd are the same.  From the same
Section~\ref{s:consistency}, the case where $i'$ is even and $j'$ odd is
the same as the case where $i'$ odd and $j'$ even, but is different from
the preceding one. So, there are essentially two cases to be
distinguished.

\subsubsection{Both $i'$ and $j'$ even or both $i'$ and $j'$ odd}

\begin{align*}
\sqrt{p^{\leftrightarrow}_{\mathrm{max}}(i',j')}
&= \frac{4}{\pi^2} \sum_{p,q} \frac{1}{pq}\int_{-1/2}^{1/2}
    \left( c(p,q)\sum \cos(2\pi(i'(p-1)\pm j'(q-1))x)\right.\\
& \quad +c(p,q)\sum \cos(2\pi(i'(p+1)\pm j'(q+1))x)\\
& \quad -d(p,q)\sum \cos(2\pi(i'(p-1)\pm j'(q+1))x)\\
& \quad \left.-d(p,q)\sum \cos(2\pi(i'(p+1)\pm j'(q-1))x)\right)dx,
\end{align*}
where
\begin{equation*}
  c(p,q) = \frac{1}{2} \left(1 +(-1)^{\frac{p+q}{2}-1}\right),   \quad 
  d(p,q) = \frac{1}{2} \left(1 +(-1)^{\frac{p+q}{2}}\right).
\end{equation*}
$c(\cdot,\cdot)$ and $d(\cdot,\cdot)$ are functions taking value $0$ or
$1$, and complementary in the sense that $c(p,q)+d(p,q)=1$. Next, we
find that
\begin{align*}
&\sqrt{p^{\leftrightarrow}_{\mathrm{max}}(i',j')}\\
&=\frac{4}{\pi^2} \sum_{p,q} \frac{1}{pq}\left(
   c(p,q)\mathbb{I}_{2\pi(i'(p-1)\pm j'(q-1))x=0}
  +c(p,q)\mathrm{I}_{2\pi(i'(p+1)\pm j'(q+1))x=0 } \right.\\
&\quad -d(p,q)\mathrm{I}_{2\pi(i'(p-1)\pm j'(q+1))x=0}
  \left.-d(p,q)\mathrm{I}_{2\pi(i'(p+1)\pm j'(q-1))x=0}\right),
\end{align*}
where
\begin{equation*}
\ip = \frac{i'}{\mathrm{gcd}(i',j')}, \quad
\jp = \frac{j'}{\mathrm{gcd}(i',j')}.
\end{equation*}
Observe that $\ip$, $\jp$ are relatively prime; hence they could not be
both even.  The developments follow closely the semi-infinite chain
case, except that, here, $\ip$ and $\jp$ are not restricted to be
positive. Hence we have to consider several cases:

\paragraph{$\ip,\jp \geq 1$}
\label{s:both_or_both_geq_1}

The following is easily observed:
\begin{align*}
\sum_{p,q=1,3,...}\frac{c(p,q)}{pq} \mathbb{I}_{\ip(p-1)-\jp(q-1)=0}(p,q)
    &=\sum_{m=0,2,...}\frac{c(m \jp +1,m\ip +1)}{(m\jp+1)(m\ip+1)},\\
\sum_{p,q=1,3,...}\frac{c(p,q)}{pq} \mathbb{I}_{\ip(p-1)+\jp(q-1)=0}(p,q)&=c(1,1),\\
\sum_{p,q=1,3,...}\frac{c(p,q)}{pq} \mathbb{I}_{\ip(p+1)+\jp(q+1)=0}(p,q)&=0,\\
\sum_{p,q=1,3,...}\frac{c(p,q)}{pq} \mathbb{I}_{\ip(p+1)-\jp(q+1)=0}(p,q)
    &=\sum_{m=2,4,...}\frac{c(m \jp -1,m \ip-1)}{(m\jp-1)(m\ip-1)}.
\end{align*}
The situation is pretty much the same for the terms involving $d(p,q)$:
\begin{align*}
\sum_{p,q=1,3,...}\frac{d(p,q)}{pq} \mathbb{I}_{\ip(p-1)-\jp(q+1)=0}(p,q)
    &=\sum_{m=2,4,...}\frac{d(m \jp +1,m\ip -1)}{(m\jp+1)(m\ip-1)},\\
\sum_{p,q=1,3,...}\frac{d(p,q)}{pq} \mathbb{I}_{\ip(p-1)+\jp(q+1)=0}(p,q)&=0,\\
\sum_{p,q=1,3,...}\frac{d(p,q)}{pq} \mathbb{I}_{\ip(p+1)+\jp(q-1)=0}(p,q)&=0,\\
\sum_{p,q=1,3,...}\frac{d(p,q)}{pq} \mathbb{I}_{\ip(p+1)-\jp(q-1)=0}(p,q)
    &=\sum_{m=2,4,...}\frac{d(m \jp -1,m \ip+1)}{(m\jp-1)(m\ip+1)}.
\end{align*}

Here we have to make a distinction between the two cases: both $i'$ and
$j'$ odd and both $i'$ and $j'$ even.  We start with the easy case where
\underline{both i' and j' are odd.} In this case indeed, both $\ip\pm
\jp$ is even. This along with $m$ is even yields
\begin{align*}
c(m\jp +1,m \ip +1) &=1,\\
c(1,1)              &=1,\\
c(m\jp -1,m \ip -1) &=1,\\
d(m\jp +1,m \ip -1) &=1,\\
d(m\jp -1,m \ip +1) &=1.
\end{align*}
Putting everything together, we find, using partial fraction decompositions,
\begin{align*}
\sqrt{p_{\mathrm{max}}^{\leftrightarrow}(i',j')}
&=\frac{4}{\pi^2}\left( 2+\sum_{m=2,4,...}\frac{4}{(m^2\ip^2-1)(m^2 \jp^2-1)} \right) \\
&=\frac{4}{\pi^2}\left( 2+\frac{4}{\ip^2-\jp^2}  \sum_{m=2,4,...} \left(  \frac{\jp^2}{m^2\jp^2-1}-\frac{\ip^2}{m^2 \ip^2-1}\right)   \right).
\end{align*}

Finally, recall (Lemma~\ref{l:cotangent}) that the infinite sums can be
expressed in terms of cotangents; this yields
\begin{equation} 
\label{e:same_parity_same_sign}
\sqrt{p_{\mathrm{max}}^\leftrightarrow (i',j')}
= \frac{8}{\pi^2} \left( \frac{1}{\ip^2-\jp^2}
  \left( 
    \ip^2 \left(\frac{\pi}{2\ip}\right)  \cot \left(\frac {\pi}{2 \ip}\right)
  - \jp^2 \left(\frac{\pi}{2\jp}\right)  \cot \left(\frac {\pi}{2 \jp}\right)  
 \right) \right).
\end{equation}

The case where \underline{both $i'$ and $j'$ are even} is more
complicated. If $i'$ and $j'$ have the same power of $2$ in their prime
number factorization, then $\ip\pm \jp$ is even and the preceding
formula holds. If the powers of $2$ are different, then $\ip\pm \jp$ is
odd and it is easily verified that
\begin{equation*} 
\left.\begin{array}{c}
c(m\jp+1,m\ip+1)\\
c(m\jp-1,m\ip-1)\\
d(m\jp+1,m\ip-1)\\
d(m\jp-1,m\ip+1)
\end{array}\right\}=
\left\{\begin{array}{cll}
1 & \mbox{for}&m=0,4,8,12,..., \\
0 & \mbox{otherwise}.&
\end{array}\right.
\end{equation*}
With the above, we get
\begin{align*}
\sqrt{p_{\mathrm{max}}^{\leftrightarrow}(\ip,\jp)}
&=\frac{4}{\pi^2}\left(2 +\sum_{m=4,8,...}
   \frac{1}{(m\jp+1)(m\ip-1)} + \frac{1}{(m\jp+1)(m\ip+1)} \right.\\
& \qquad\qquad -\left.\frac{1}{(m\jp'+1)()m\ip-1}-\frac{1}{(m\jp-1)(m\ip+1)}\right)\\
&= \frac{4}{\pi^2}\left( 2+\sum_{m=4,8,...} \frac{4}{(m^2\jp^2-1)(m^2\ip^2-1)} \right)\\
&= \frac{4}{\pi^2}\left(  2+ \frac{4}{\ip^2-\jp^2}\sum_{m=4,8,...}\left(  \frac{\jp^2}{m^2\jp^2-1}-\frac{\ip^2}{m^2\ip^2-1}\right)\right).
\end{align*}

In order to derive a closed-form representation of the infinite series,
we need the following lemma:
\begin{lemma}
\label{l:4_cotangent}
\[ \sum_{m=4,8,...}\frac{1}{m^2\ip^2-1}=
\frac{1}{2}\left( 1-\left(\frac{\pi}{4\ip}\right)\cot\left( \frac{\pi}{4\ip}\right)  \right). \]
\end{lemma}

\begin{proof}
The proof is the same as that of Lemma~\ref{l:cotangent}, except that
instead of setting $z=1/2{\bf i}$ we set $z=1/4\ip$.
\end{proof}

Using the lemma, we finally get the closed-form formula:
\begin{equation*} 
 \sqrt{p_{\mathrm{max}}^{\leftrightarrow}(\ip,\jp)}=
\frac{8}{\pi^2} \left(
\frac{\ip^2}{\ip^2-\jp^2}\left( \frac{\pi}{4\ip}\right) \cot \left( \frac{\pi}{4\ip}\right)
-\frac{\jp^2}{\ip^2-\jp^2}\left( \frac{\pi}{4\jp}\right) \cot \left( \frac{\pi}{4\jp}\right)
\right).
\end{equation*}

\paragraph{$\ip < 0 < \jp$}

The preceding formula remains valid, after replacing $i'$ by $-i'$. 

\subsubsection{$i'$ odd  and $j'$ even}

From the integral representation, we get
\begin{align*}
&\sqrt{p^{\leftrightarrow}_{\mathrm{max}}(i',j')}\\
&=\frac{4}{\pi^2} \sum_{p,q} \frac{1}{pq}\int_{-1/2}^{1/2}\left( (-1)^{\frac{p-1}{2}}c(p,q)\sum \cos(2\pi(i'(p-1)\pm j'(q-1))x)\right.\\
& \qquad -(-1)^{\frac{p-1}{2}}c(p,q)\sum \cos(2\pi(i'(p+1)\pm j'(q+1))x)\\
& \qquad -(-1)^{\frac{p-1}{2}}d(p,q)\sum \cos(2\pi(i'(p-1)\pm j'(q+1))x)\\
& \qquad +(-1)^{\frac{p-1}{2}}\left.d(p,q)\sum \cos(2\pi(i'(p+1)\pm j'(q-1))x)\right)dx.
\end{align*}
Next, we find that 
\begin{align*}
\sqrt{p^{\leftrightarrow}_{\mathrm{max}}(i',j')}
&=\frac{4}{\pi^2} \sum_{p,q} \frac{1}{pq}
  \left( (-1)^{\frac{p-1}{2}}c(p,q)\mathbb{I}_{2\pi(i'(p-1)\pm j'(q-1))x=0}\right.\\
& \qquad -(-1)^{\frac{p-1}{2}}c(p,q)\mathrm{I}_{2\pi(i'(p+1)\pm j'(q+1))x=0}\\
& \qquad -(-1)^{\frac{p-1}{2}}d(p,q)\mathrm{I}_{2\pi(i'(p-1)\pm j'(q+1))x=0}\\
& \qquad \left.+(-1)^{\frac{p-1}{2}}d(p,q)\mathrm{I}_{2\pi(i'(p+1)\pm j'(q-1))x=0}\right).
\end{align*}

Despite the extra difficulties created by the $c(\cdot,\cdot)$ and
$d(\cdot,\cdot)$ functions and the various signs, the pattern remains
the same as before: the indicators are nonvanishing only if
\[ p=m\jp \pm 1, \quad q=\ip \pm 1 \]
for some even $m$. 

\paragraph{$\ip, \jp>0$}

Since $i'$ is odd and $j'$ is even, $\mathrm{gcd}(i',j')$ does not
contain any positive power of $2$ in its prime divisors; therefore,
$\ip$ remains odd and $\jp$ remains even; in other words, $\ip+\jp$ is
odd.  From this observation, tedious but elementary manipulations lead
to the following:
\[ \left.\begin{array}{c}
c(m\jp+1,m\ip+1)\\
c(m\jp-1,m\ip-1)\\
d(m\jp+1,m\ip-1)\\
d(m\jp-1,m\ip+1)
\end{array}\right\}=
\left\{\begin{array}{cll}
1 & \mbox{for}&m=0,4,8,12,..., \\
0 & \mbox{otherwise}.&
\end{array}\right. \]

Next, tedious but elementary manipulation reveal that
\[ (-1)^{\frac{p-1}{2}}=
\left\{\begin{array}{ccc}
1 & \mbox{if} & p=m\jp + 1,\\
0 & \mbox{if} & p=m\jp - 1.
\end{array}\right.\]

Putting everything together yields
\begin{align*}
\sqrt{p_{\mathrm{max}}^{\leftrightarrow}(\ip,\jp)}
&=\frac{4}{\pi^2}\left(2 +\sum_{m=4,8,...}
   \frac{1}{(m\jp+1)(m\ip-1)} + \frac{1}{(m\jp+1)(m\ip+1)} \right. \\
& \qquad\qquad -\left.\frac{1}{(m\jp'+1)()m\ip-1}-\frac{1}{(m\jp-1)(m\ip+1)}\right)\\
&=\frac{4}{\pi^2}\left( 2+\sum_{m=4,8,...} \frac{4}{(m^2\jp^2-1)(m^2\ip^2-1)} \right)\\
&=\frac{4}{\pi^2}\left(  2+ \frac{4}{\ip^2-\jp^2}\sum_{m=4,8,...}\left(  \frac{\jp^2}{m^2\jp^2-1}-\frac{\ip^2}{m^2\ip^2-1}\right)\right).
\end{align*}
In order to derive the closed form solution, we invoke
Lemma~\ref{l:4_cotangent} and find that
\begin{equation*} 
 \sqrt{p_{\mathrm{max}}^{\leftrightarrow}(\ip,\jp)}=
\frac{8}{\pi^2} \left(
\frac{\ip^2}{\ip^2-\jp^2}\left( \frac{\pi}{4\ip}\right) \cot \left( \frac{\pi}{4\ip}\right)
-\frac{\jp^2}{\ip^2-\jp^2}\left( \frac{\pi}{4\jp}\right) \cot \left( \frac{\pi}{4\jp}\right)
\right).
\end{equation*}

\paragraph{$\ip < 0 <\jp$}

Again the preceding formula remains valid after replacing $i'$ by $-i'$. 

\bibliographystyle{plain}

\end{document}